\pdfoutput=1
\documentclass[11pt]{article}

\PassOptionsToPackage{svgnames}{xcolor}
\usepackage[preprint]{acl}
\usepackage{microtype}
\usepackage{amsmath,amssymb,mathtools,amsthm}
\usepackage{float}
\floatstyle{ruled}
\newfloat{algorithm}{tbp}{loa}
\floatname{algorithm}{Algorithm}
\floatstyle{plain}

\usepackage{booktabs}
\usepackage{colortbl}
\definecolor{TableHead}{HTML}{EAF0F5}
\definecolor{TableAccent}{HTML}{F2F6FA}
\newcommand{\tablestandard}{\setlength{\tabcolsep}{4pt}\renewcommand{\arraystretch}{1.16}}
\newcommand{\tablecompact}{\setlength{\tabcolsep}{3pt}\renewcommand{\arraystretch}{1.16}}

\usepackage{needspace,placeins}
\usepackage{enumitem}
\usepackage{graphicx}
\usepackage{listings}
\lstdefinestyle{paperalgorithm}{
  basicstyle=\normalsize\ttfamily,
  mathescape=true, breaklines=true, columns=fullflexible,
  keepspaces=true, showstringspaces=false,
  aboveskip=0.3\baselineskip, belowskip=0pt
}
\hypersetup{
  colorlinks=true,
  linkcolor=MidnightBlue,
  urlcolor=MidnightBlue,
  citecolor=MidnightBlue
}

\title{Reflective Compositional Grounding for Natural-Language-to-SPARQL}
\author{Isabelle Mohr \\
  Idiap Research Institute \\
  EPFL \\
  \texttt{isabelle.mohr@idiap.ch} \\\And
  John Dujany \\
  Merck KGaA \\
  \texttt{john.dujany@merck.com} \\\And
  Andre Freitas \\
  Idiap Research Institute \\
  \texttt{andre.freitas@idiap.ch} \\}
\date{}

\newcommand{\KG}{\mathcal{K}}
\newcommand{\Schema}{\Sigma}

\newcommand{\Enc}{\operatorname{Enc}}

\newtheorem{definition}{Definition}
\newtheorem{proposition}{Proposition}
\newtheorem{theorem}{Theorem}

\begin{document}
\maketitle
\begin{abstract}
Natural-language query generation over large knowledge graphs embodies the classical vocabulary problem in human-systems communication. Questions express concepts through wording and compositions that can differ substantially from the target schema, while limited model context requires selective access to schema information. We present an NL2SPARQL framework that places conceptual compositionality at the center of query interpretation and connects it to the relational composition required for binding and query construction. A fixed concept encoder represents compound expressions jointly, and schema candidates are ranked using semantic similarity, type compatibility, and hierarchy proximity. Reflective refinement critically examines how the grounded meanings are realized through argument bindings, joins, and operator scope. Diagnosed semantic discrepancies guide revisions to the responsible stage. We formalise a monotone refinement discipline that accumulates scoped interpretation requirements and establish conditions under which changes to the composed process preserve previously validated interpretations. Across five benchmark/knowledge-graph targets, mean schema recall@10 increases from 0.712 to 0.837. With retrieval held fixed, reflective prompt rewriting adds 2.27-3.33 answer-F1 points. These experiments evaluate the grounding and refinement mechanisms, while the formal analysis specifies the stronger requirements for conservative adaptation across queries.
\end{abstract}

\section{Introduction}

\begin{figure*}[t]
\centering
\resizebox{\textwidth}{!}{%
\input{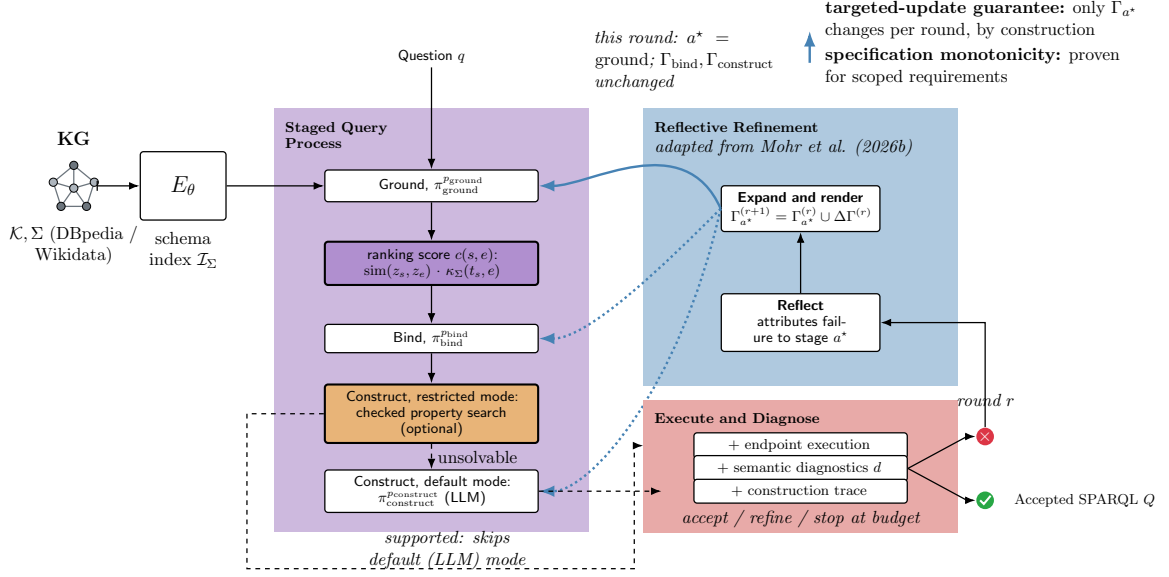}\unskip%
}
\caption{\textbf{Reflective refinement and restricted symbolic construction extend the staged query process.} Ground ranks schema candidates by $c(s,e)$ using the fixed encoder $E_\theta$; Bind and Construct carry them into an executable query, with an optional restricted symbolic mode checked against the graph before the default LLM mode. On a diagnosed violation, Reflect attributes it to one stage $a^\star$ and extends only that stage's requirement set $\Gamma_{a^\star}$ (the figure's targeted-update guarantee); specification monotonicity is proven in Appendix~\ref{app:formalisms-refinement}.}
\label{fig:hero}
\end{figure*}

Natural-language access to databases faces a vocabulary problem. Users express information needs through concepts and relations that can differ in wording and conceptual structure from the target schema \citep{freitas2015thesis,freitas2012terminological}, a gap that widens in large, heterogeneous knowledge graphs with overlapping concepts and contextual ambiguities \citep{freitas2015complexity}. Large schemas can also exceed a language model's available context, making selective, structurally-aware retrieval necessary, which is itself a form of \textit{conceptual interpretation at scale}, since relevant schema elements may share little lexical overlap with the question. Query generation must therefore construct a relevant schema context, ground the question's concepts within it, and compose their interpretations into an executable query; \textit{conceptual interpretation} and \textit{compositionality} are thus central to the end-to-end mapping.

Consider the question ``\textit{Which organisations own stakes in unlisted companies headquartered in the same country as themselves?}'' In an illustrative investment schema, \texttt{PrivateEquity} denotes equity holdings in unlisted companies; interpreting ``stakes in unlisted companies'' requires recognising this concept despite limited lexical overlap, and the country constraint requires connecting each organisation's headquarters to those of the company in which it holds a stake. The first decision concerns meaning within a compound expression, the second how that interpretation combines with ownership and location constraints, both of which the retrieved schema context must support. Appendix~\ref{app:mapping-difficulty} provides the schema and complete query.

\paragraph{Phenomena.}
Different forms of composition impose different interpretation requirements: ``French scientists'' combines nationality and occupation, ``former presidents'' requires interpreting a previous role, and relational expressions determine which entities occupy which roles and which constraints they share. These distinctions can be lost at any stage (schema selection, concept grounding, argument binding, or query construction), thus allowing a query to execute successfully while expressing a different information need. Correcting such failures requires locating the responsible decision and controlling how its revision affects the surrounding composition.

\paragraph{Mechanism I (Conceptual composition and mapping).}
We place compound expressions at the centre of schema grounding via a \textit{fixed concept encoder} that represents each expression and candidate schema description in a shared semantic space. For the running example, retrieval encodes ``stakes in unlisted companies'' jointly, making the compound expression the unit of comparison with the schema (Figure~\ref{fig:hero}); candidate ranking combines embedding similarity with directional type compatibility and hierarchy proximity. The resulting candidates pass to binding and construction, where entities, values, shared variables, and operators determine their contribution to the complete query; this connects composition within the encoder's input to composition across the query process.

\paragraph{Mechanism II (Reflective refinement of relational composition).}
Reflection addresses how grounded meanings compose through bindings and structural relations, critically interpreting the candidate query against the question and schema context to examine whether argument roles, modifier attachments, joins, and operator scope jointly express the intended semantics. In the running example, the country comparison must bind the organisation to the company reached through its ownership relation. A diagnosed discrepancy becomes a scoped requirement on the responsible stage; admitted refinements add these requirements while preserving earlier obligations, monotonically strengthening the prompt specification, and the affected stages are recomputed under renewed critique. This complements conceptual compositionality in representation with relational compositionality in construction, using semantic critique to refine how grounded concepts participate in the complete query.

\paragraph{Contributions.}
We provide an end-to-end formalisation connecting compositional representations, contextual interpretation mappings $\mu$, and the staged query process $\pi$ that applies them, stating conditions for identifying mappings from diagnostic evidence and proving a conditional retention theorem for updates that preserve protected mapping meanings and satisfy compatibility contracts across stage interfaces. Empirically, we transfer an existing concept encoder into schema retrieval and introduce hierarchy-guided ranking. Across five benchmark/knowledge-graph targets, mean recall@10 increases from 0.712 to 0.837. With retrieval held fixed, reflective prompt rewriting adds 2.27--3.33 answer-F1 points. Separate ablations distinguish retrieval, reflection, and restricted symbolic construction, while the diagnostic benchmark examines variation across composition-related phenomena.

\paragraph{Reproducibility and scope.}
Appendix~\ref{app:transition} records the evaluation configuration, Appendix~\ref{app:diagnostic-construction} describes the diagnostic benchmark, and Appendices~\ref{app:selection-criteria} and~\ref{app:base-models} document selection criteria and model characteristics, identifying the run artifacts and provenance needed for independent reproduction.

\section{Proposed Approach}
\label{sec:method}

We model query interpretation through two complementary forms of composition. Conceptual composition determines an expression's meaning for schema retrieval, while relational composition determines how concepts participate in bindings, joins, and operators. In ``French scientists working at universities in Paris,'' this means distinguishing the scientists' nationality from their employers' location. Grounding, binding, and construction realize these decisions, with reflective refinement revising each stage's requirements.

\subsection{Staged Query Process}
\label{sec:mappings}
\label{sec:method-implementation}

Given a question $q$, a knowledge graph $\KG$, and its schema $\Sigma$ (external resources accessed through retrieval and lookup), the system constructs a SPARQL query $Q$. Let $\mathcal I_\Sigma$ denote an index of schema concepts encoded by a fixed encoder $E_\theta$. The query process is
\begin{equation}
\begin{aligned}
\pi_{\mathbf p,E_\theta,\mathcal I_\Sigma}
={}&\pi_{\mathrm{construct}}^{p_{\mathrm{construct}}}
\circ\pi_{\mathrm{bind}}^{p_{\mathrm{bind}}}\\
&\circ\pi_{\mathrm{ground}}^{
p_{\mathrm{ground}},E_\theta,\mathcal I_\Sigma}.
\end{aligned}
\label{eq:pi-composition}
\end{equation}

Grounding identifies expressions and their roles, retrieves candidate schema interpretations, and assembles the relevant schema context. Binding associates these interpretations with entities, literals, and constraints. Construction combines the bindings into an executable query. Stage states retain the question, prior decisions, and their provenance, and each stage receives only the portion it needs.

These responsibilities correspond to the schema-recall, value-extraction, and query-synthesis stages in Figure~\ref{fig:hero}. Appendix~\ref{app:approach-details} specifies their accumulated states and interfaces.

\subsection{Compositional Encoding and Nearest-Neighbour Retrieval}
\label{sec:composition}
\label{sec:encoder}
\label{sec:type-retrieval}

We extend \citet{mohr2026concepts}'s compositional concept encoder with a hierarchical type-constrained retrieval score for schema grounding. The encoder  $E_\theta$ remains fixed, providing a similarity signal for the retrieval step. Let $\mathcal V_\Sigma$ be the searchable vocabulary of classes and properties, and let $d_\Sigma(e)$ be the textual description used to index schema element $e$. The schema index stores
\begin{equation}
z_e=E_\theta(d_\Sigma(e)),
\qquad
\mathcal I_\Sigma=\{(e,z_e):e\in\mathcal V_\Sigma\}.
\label{eq:schema-index}
\end{equation}

The grounding prompt identifies expression occurrences
$\mathcal S_q$, retaining their roles and question context.
Each compound expression $s$ is encoded jointly as
$z_s=E_\theta(s)$.

Retrieval first obtains a pool of $K$ neighbours in the embedding
space. Structural ranking then retains $k\leq K$ candidates,
\begin{align}
\mathcal N_K(s)
&=\operatorname{kNN}_K(z_s;\mathcal I_\Sigma),
\label{eq:knn-pool}\\
c(s,e)
&=\operatorname{sim}(z_s,z_e)\,
\kappa_\Sigma(t_s,e),
\label{eq:c-score}\\
\mathcal C_k(s)
&=\operatorname{Top}_{k,\,e\in\mathcal N_K(s)}c(s,e).
\label{eq:retrieved-candidates}
\end{align}

The nearest-neighbour search uses the embedding similarity
$\operatorname{sim}$, whose exact or approximate implementation is
a backend choice. The contextual type $t_s$ supplies the anchor
for structural ranking. The factor $\kappa_\Sigma$ combines
directional type compatibility with hierarchy proximity, using
a class candidate's hierarchy position or a property's domain.
Unresolved types leave the embedding similarity unchanged.

The encoder determines the semantic neighbourhood available to
grounding. Structural ranking selects within that neighbourhood
and cannot recover an element absent from the retrieved pool.
RH1 evaluates encoder transfer, and RH2 evaluates the additional
effect of structural ranking. Figure~\ref{fig:mechanism} illustrates
the score. Appendix~\ref{app:approach-details} defines the structural
factor and its metadata policies.

\subsection{Schema Context and Stage Inputs}
\label{sec:context-management}

The retrieved lists define a query-specific schema view
$\widehat\Sigma_q$, retaining for each expression its candidates'
identifiers, descriptions, and relevant type and relation
declarations, so later stages can distinguish which interpretation
a candidate supports and how it participates in the query.

The complete schema index and accumulated diagnostic records
remain external to the language-model context. Each stage receives
the question, the relevant schema view, and the predecessor
decisions needed for its task. Context selection and serialization
are part of the stage interface and operate within its input
budget. Any further candidate pruning must be recorded because it
changes the interpretations available to subsequent stages.

Refinement requirements are also selected by applicability.
For an incoming stage state $x$, define
\begin{equation}
\Gamma_a^{(r)}[x]
=
\{\delta\in\Gamma_a^{(r)}:g_\delta(x)\},
\label{eq:active-requirements}
\end{equation}
where $g_\delta$ is the requirement's applicability guard.
The stage prompt is instantiated as
\begin{equation}
p_a^{(r)}(x)
=
\operatorname{Render}
\bigl(B_a,\Gamma_a^{(r)}[x]\bigr).
\label{eq:prompt-render}
\end{equation}
Here $B_a$ contains the fixed stage instructions.
Thus $p_a^{(r)}$ denotes a prompt policy instantiated on the current
state. Requirements accumulate in the stored specification,
while each invocation receives its applicable obligations.
Faithful rendering and context management must preserve those
obligations within the available input budget.

\subsection{Relational Binding and Construction}
\label{sec:solver}

Binding determines which entities, values, and constraints
instantiate the retrieved concepts and relations. Construction
determines their argument positions, shared variables, joins,
and operator scope. Together, these stages realize the relational
composition of the grounded meanings.

In the opening example, retrieval supplies alternatives for
constraining the scientist's nationality, connecting them to their
university via employment, and constraining that university's
location. Binding and construction determine whether the
combination expresses the question.

The default constructor generates SPARQL with an LLM.
An optional symbolic constructor searches restricted property
assignments and checks their support against the graph.
Its search space, fallback behaviour, and conditional guarantees
are specified in Appendix~\ref{app:formalisms-solver}.
Ablation A7 evaluates this constructor separately.

\subsection{Reflective Refinement}
\label{sec:refinement}

We extend \citet{mohr2026reflective}'s reflective-refinement framework for SQL construction to the grounding, binding and construction stages of SPARQL. Our mechanism uses execution results and semantic critique to examine the current query, assessing whether the retrieved meanings were realized through the intended bindings and structural relations, and interpreting diagnostics against the question, the supplied schema context, and the construction trace.

A failure can arise from the expression submitted to retrieval,
the candidates made available, an argument binding, or the
composition of the final query. Reflection attributes the
discrepancy to a stage and proposes a scoped requirement governing
its correction, recording the contexts in which that requirement
applies.

For the attributed stage $a^\star$, an admitted update gives
\begin{equation}
\Gamma_a^{(r+1)}=
\begin{cases}
\Gamma_a^{(r)}\cup\Delta\Gamma^{(r)},&a=a^\star,\\
\Gamma_a^{(r)},&a\neq a^\star.
\end{cases}
\label{eq:requirement-update}
\end{equation}
Admission requires consistency on the declared scope and faithful
rendering of the applicable requirements. For each fixed stage
input, every behaviour satisfying the strengthened specification
therefore satisfies its earlier obligations.
Appendix~\ref{app:formalisms-refinement} establishes this monotonic
specification property.

The affected stage and its successors are recomputed.
A grounding revision can change the extracted expressions,
their embeddings, the retrieved neighbours, and the schema context,
although the encoder and index remain fixed. Binding and
construction are then evaluated on the revised inputs.
Fresh diagnostics assess the resulting query.

The retention conditions below cover this complete composition,
including retrieval and context selection. Requirement
accumulation does not require nested candidate lists.
Preservation concerns the validated interpretations supported
through the stage interfaces. RH3 evaluates the answer-quality effect of prompt rewriting, while structured admission and cross-query retention remain formal extensions.

\subsection{Conservative Inference-Time Learning}
\label{sec:inference-time-learning}

The reflective process extends across questions via a protected bank $M_t$ of interpretation mappings $\mu$, each with a declared applicability scope and provenance and marked as either validated or a revisable hypothesis; the staged procedure $\pi_t$ consults this bank during grounding, binding, and construction, with stage prompts that can change while the encoder $E_\theta$ remains fixed (Appendix~\ref{app:conservative-learning}). A protected history $H_t$ records validated interpretations together with their contexts, versions, and traces; a new question can propose a mapping, extend its scope, or revise a stage prompt, but admission must preserve protected interpretations while allowing unvalidated hypotheses to change, since a new rule can shadow or recompose with an earlier one even if the bank still contains it.

Three conditions support identification of a mapping: semantic validation distinguishing the intended interpretation up to equivalence, context separability distinguishing cases needing different interpretations, and diagnostic traces justifying the proposed mapping, stage, and scope without requiring a unique repair (Appendix~\ref{app:conservative-learning}).

\begin{theorem}[Conservative retention]
\label{thm:conservative-learning}
Suppose an update preserves protected mapping meanings and satisfies total initialization and local simulation contracts across the grounding, binding, and construction interfaces. If the final interface preserves semantic interpretation, every covered run on the protected history retains its validated meaning,
\[
\mathcal J_{t+1}(q_h,c_h;\tau)\simeq_{c_h}
I_h,\qquad h\in H_t.
\]
\end{theorem}

Here $I_h$ is the interpretation certified in the historical trace $\tau_h$, and $\tau$ ranges over covered new runs; the proof, given in Appendix~\ref{app:conservative-learning}, requires an unchanged stage to respect the interface relation under changed inputs and a changed stage to establish the corresponding local compatibility condition, with contracts covering termination and all permitted outputs rather than a single successful sample.

This retention property is distinct from monotone prompt requirements. Inclusion of $\Gamma_a^{(r)}$ preserves stated obligations but does not establish semantic compatibility of a changed process. A finite protected history also supports a narrower claim than preservation over an entire domain. This is a conceptual and specification extension of the approach; the existing experiments do not evaluate persistent mapping memory or cross-query retention.

\section{Empirical Analysis}
\label{sec:empirical}

\subsection{Research Hypotheses}

\paragraph{RH1 (Compositional retrieval).}
The composition-aware encoder $E_\theta$ improves schema recall@10 over the baseline text encoder under matched retrieval settings.

\paragraph{RH2 (Structural grounding).}
Ranking compositional retrieval candidates by type compatibility and hierarchy proximity improves schema recall@10 beyond embedding similarity.

\paragraph{RH3 (Reflective refinement).}
Stage-attributed reflective prompt refinement improves answer F1 over the same query-generation pipeline with refinement disabled.

\paragraph{RH4 (Consistency across backbones).}
Compositional retrieval with structural ranking improves answer F1 over baseline retrieval within each evaluated language-model backbone.

\begin{table*}[!t]
\centering
\footnotesize
\tablecompact
\begin{tabular}{@{}lrrrrrrrrrrr@{}}
\toprule
 & \multicolumn{2}{c}{Recall@5} & \multicolumn{3}{c}{Recall@10} & \multicolumn{2}{c}{Recall@20} & \multicolumn{3}{c}{Answer MF1 (\%)} & \\
\cmidrule(lr){2-3}\cmidrule(lr){4-6}\cmidrule(lr){7-8}\cmidrule(lr){9-11}
\rowcolor{TableHead}
Benchmark/KG target & Flat & Full & Flat & Encoder & Full & Flat & Full & Flat & Encoder & Full & $\Delta$MF1 (pp) \\
\midrule
Spider4SPARQL & 0.652 & 0.791 & 0.745 & 0.821 & 0.859 & 0.831 & 0.912 & 80.15 & 83.41 & 85.67 & +5.52 \\
QALD-9-Plus (DBpedia) & 0.601 & 0.748 & 0.698 & 0.784 & 0.827 & 0.789 & 0.887 & 67.53 & 71.28 & 74.06 & \textbf{+6.53} \\
QALD-9-Plus (Wikidata) & 0.578 & 0.721 & 0.671 & 0.759 & 0.801 & 0.768 & 0.869 & 63.98 & 67.12 & 69.84 & +5.86 \\
QALD-10 (static) & 0.589 & 0.735 & 0.683 & 0.772 & 0.813 & 0.781 & 0.878 & 55.71 & 59.03 & 61.87 & +6.16 \\
LC-QuAD 2.0 & 0.671 & 0.812 & 0.763 & 0.844 & 0.885 & 0.849 & 0.896 & 58.42 & 62.15 & 64.91 & +6.49 \\
\midrule
\rowcolor{TableAccent}
Mean across targets & 0.618 & 0.761 & 0.712 & 0.796 & \textbf{0.837} & 0.804 & 0.888 & 65.16 & 68.60 & 71.27 & +6.11 \\
\bottomrule
\end{tabular}
\caption{\textbf{Full retrieval improves recall and answer quality on all five targets.}
RH1--RH2: \emph{Encoder} replaces flat representations with $\Enc_\theta$; \emph{Full} adds type/hierarchy ranking (refinement off).
Recall@$k$ is a fraction; the Full--Flat gap holds at every depth (largest at $k=5$, still positive at $k=20$).
$\Delta$MF1 is Full minus Flat (pp). Means are unweighted. Bold marks the largest MF1 gain and the final mean recall.
External comparisons appear in Table~\ref{tab:external-comparators}.}
\label{tab:baseline-results}
\end{table*}

\subsection{Benchmarks}
\label{sec:diagnostic-benchmark}

We evaluate against four established benchmark families (yielding five benchmark/KG targets, since QALD-9-Plus is evaluated on two graphs) and one diagnostic benchmark: QALD-9-Plus \citep{perevalov2022qald9plus} on DBpedia and Wikidata, QALD-10 \citep{usbeck2024qald10} on Wikidata with both live and static reference conditions, LC-QuAD 2.0 \citep{dubey2019lcquad2} at larger scale, Spider4SPARQL \citep{kosten2023spider4sparql} across heterogeneous per-domain schemas, and our own composition-type diagnostic set. Appendix~\ref{app:benchmark-selection} (Table~\ref{tab:selection-criteria}) records each target's recorded question count, scored coverage, and inclusion/exclusion criteria; Appendix~\ref{app:diagnostic-construction} describes the diagnostic set's construction; Appendix~\ref{app:transition} specifies the evaluation configurations and diagnostics.

\subsection{Baselines}

The controlled comparisons use three retrieval configurations: \emph{Flat} (baseline sentence encoder), \emph{Encoder} (the fixed composition-aware encoder), and \emph{Full} (adding type compatibility and hierarchy ranking), with the same generation configuration and reflective refinement disabled. The reflection comparison then enables stage-attributed prompt refinement under Full retrieval, and the backbone comparison repeats Flat--Full within each evaluated model. Published results for QAnswer, SPINACH, mKGQAgent, SchemaForge, and SAGA provide additional context, though these systems were not rerun. Appendix~\ref{app:transition} specifies the comparison controls; Appendix~\ref{app:baseline-selection} documents the external systems, and numerical results appear in Appendix~\ref{app:additional-results}.

\subsection{Ablations}

We retain the experiment identifiers used by the result tables. A1 compares Encoder with Flat (RH1), A2 compares Full with Encoder (RH2), A4 enables stage-attributed prompt rewriting under Full retrieval (RH3), and A6 repeats the Flat--Full comparison across four backbones on QALD-10 and QALD-9-Plus DBpedia (RH4). Appendix~\ref{app:base-models} documents model selection and recorded configuration scope.

Supporting analyses examine composition-related strata and the shared subset with correct entities and query structure, via A3 (gold entities) and A7 (restricted graph-checked construction); the unreconciled A5 budget summary is excluded (Appendix~\ref{app:baseline-selection}).

\subsection{Metrics and Grounding Diagnostics}
\label{sec:diagnostic}

We report answer F1, exact match or execution accuracy under each table's stated convention, and schema-element recall@$k$. For the solver we also report its fire rate, the fraction of questions with both a supported shape and an accepted assignment.

Let $I$ be the intersection of questions with correct entity linking and query structure under both compared configurations. Two diagnostics must be distinguished:
\begin{align}
\Delta_{\mathrm{ground}}&=A_{\mathrm{full}}-A_{\mathrm{flat}},\nonumber\\
\Delta_{\mathrm{F1}\mid I}&=F_{\mathrm{full}}-F_{\mathrm{flat}},
\label{eq:diagnostics}
\end{align}
where $A_c=\operatorname{Acc}_{\mathrm{ground}}(c\mid I)$ and $F_c=\operatorname{F1}(c\mid I)$.
Grounding accuracy requires the joint class-and-property assignment to match its reference. Answer F1 measures the returned answer set. Table~\ref{tab:conditioned-answers} supplies the second quantity only. Conditioning reduces variation from the named upstream errors but does not identify a causal effect of retrieval or establish that the two diagnostics coincide. Appendix~\ref{app:diagnostic-derivation} states the reporting requirements.

\begin{table*}[!t]
\centering
\small
\tablestandard
\begin{tabular}{@{}llrrrrrr@{}}
\toprule
 & & \multicolumn{3}{c}{QALD-10 (static), F1 (\%)} & \multicolumn{3}{c}{\shortstack{QALD-9-Plus (DBpedia)\\F1 (\%)}} \\
\cmidrule(lr){3-5}\cmidrule(lr){6-8}
\rowcolor{TableHead}
Backbone & Architecture / access & Flat & Full & $\Delta$ (pp) & Flat & Full & $\Delta$ (pp) \\
\midrule
Qwen2.5-7B-Instruct & Dense / weights & 31.48 & 39.72 & \textbf{+8.24} & 38.91 & 48.64 & \textbf{+9.73} \\
Qwen3-32B & Dense / weights & 44.17 & 51.03 & +6.86 & 52.36 & 59.81 & +7.45 \\
gpt-oss-120b & MoE / weights & 52.84 & 58.19 & +5.35 & 61.07 & 67.42 & +6.35 \\
gpt-5.5 & Undisclosed / API & 55.71 & 61.87 & +6.16 & 67.53 & 74.06 & +6.53 \\
\bottomrule
\end{tabular}
\caption{\textbf{Retrieval improves answer F1 with all four backbones.}
Backbone F1 is unconditioned (cf.\ Table~\ref{tab:conditioned-answers}).
Bold marks the largest gain per target. Comparisons do not isolate parameter count; see Appendix~\ref{app:base-models}.}
\label{tab:prospective-backbone}
\end{table*}

\subsection{Results}
\label{sec:prospective-results}
\label{sec:baseline-results}

\begin{table*}[!t]
\centering
\begin{minipage}[t]{0.48\textwidth}
\vspace{0pt}
\centering
\scriptsize
\tablecompact
\begin{tabular}{@{}llrrrr@{}}
\toprule
\rowcolor{TableHead}
Abl. & Metric & Base & Added & $\Delta$ & 95\% interval \\
\midrule
A1 & Recall@10 & 0.712 & 0.796 & +0.084 & $[0.063,0.105]$ \\
A2 & Recall@10 & 0.796 & 0.837 & +0.041 & $[0.024,0.057]$ \\
\addlinespace[3pt]
A3 & F1 (\%) & 64.91 & 71.38 & +6.47 & $[5.03,7.90]$ \\
\rowcolor{TableAccent}
A4 & F1 (\%) & 64.91 & 68.24 & \textbf{+3.33} & $[1.90,4.76]$ \\
\midrule
A6 & F1 & \multicolumn{4}{l}{Backbones: Table~\ref{tab:prospective-backbone}} \\
A7 & F1, EA & \multicolumn{4}{l}{Solver: Table~\ref{tab:solver-results}} \\
\bottomrule
\end{tabular}
\caption{\textbf{Ablations separate retrieval gains from prompt-refinement gains.}
Recall@10 means (A1, A2) are equally weighted across five targets; A4's F1 baseline is the full-retrieval LC-QuAD 2.0 configuration.
Recall deltas are fractions, F1 deltas are pp. Paired-bootstrap 95\% CIs reflect within-run question sampling.}
\label{tab:prospective-ablations}

\end{minipage}\hfill
\begin{minipage}[t]{0.48\textwidth}
\vspace{0pt}
\centering
\footnotesize
\tablecompact
\begin{tabular}{@{}lrrr@{}}
\toprule
 & \multicolumn{2}{c}{Answer F1 (\%)} & \\
\cmidrule(lr){2-3}
\rowcolor{TableHead}
Benchmark/KG & No refine. & Refine. & $\Delta$ (pp) \\
\midrule
Spider4SPARQL & 85.67 & 87.94 & +2.27 \\
QALD-9-Plus (DBp.) & 74.06 & 76.81 & +2.75 \\
QALD-9-Plus (Wiki.) & 69.84 & 72.53 & +2.69 \\
QALD-10 & 61.87 & 64.72 & +2.85 \\
\rowcolor{TableAccent}
LC-QuAD 2.0 & 64.91 & 68.24 & \textbf{+3.33} \\
\bottomrule
\end{tabular}
\caption{\textbf{Prompt rewriting improves answer F1 on every target.}
$\Delta$ is refinement minus no-refinement (pp); largest gain highlighted.}
\label{tab:a4-by-benchmark}

\end{minipage}
\end{table*}

Compositional encoding improves retrieval on all five targets.
Replacing the baseline encoder raises mean schema recall@10 from 0.712 to 0.796, an increase of 8.4 percentage points (RH1, Table~\ref{tab:baseline-results}). Per-target gains range from 7.6 to 8.9 points. A1 compares the encoder choices within the reported retrieval configuration.

Schema structure adds gains beyond compositional encoding.
Type compatibility and hierarchy ranking raise mean recall@10 from 0.796 to 0.837, adding 4.1 percentage points (RH2). Every target improves, with gains of 3.8--4.3 points. This comparison evaluates the combined structural score and leaves the separate effects of compatibility and hierarchy distance unresolved.

The full retrieval configuration also improves unconditioned answer MF1 over Flat by 5.52--6.53 points on every target, with reflective refinement disabled. Table~\ref{tab:prospective-ablations} reports the component comparisons, and Appendix~\ref{app:additional-results} plots the retrieval pattern and records external comparator results.

\label{sec:refinement-results}
Reflective refinement improves answers on every target.
With the Full retrieval configuration fixed, stage-attributed prompt rewriting adds 2.27--3.33 answer-F1 points across the five targets (RH3, Table~\ref{tab:a4-by-benchmark}). Reflection provides a further improvement after compositional retrieval and structural ranking. These experiments evaluate the implemented prompt-rewriting mechanism. They do not certify monotone requirement accumulation or semantic retention across queries. Appendix~\ref{app:refinement-budget-sweep} reports how this gain varies with the refinement-round budget.
\noindent\textbf{Retrieval helps all four tested backbones.}
Every within-backbone comparison is positive on both targets (RH4, Table~\ref{tab:prospective-backbone}). Gains range from 5.35 to 9.73 F1 points, with the largest for Qwen2.5-7B-Instruct. The highest-scoring tested backbone gains 6.16 points on QALD-10 and 6.53 points on QALD-9-Plus DBpedia.

Appendix~\ref{app:supporting-diagnostics} reports four supporting diagnostics examining where the gains occur and the limits of the construction mechanism: the largest retrieval gain concentrates in the subsective stratum; answer gains persist after conditioning on correct entities and query structure; reflective refinement adds further gains even with gold entities supplied; and replacing default construction with the restricted solver loses pooled accuracy despite gains in two strata.

\section{Related Work}
\label{sec:relwork}

Appendix~\ref{app:extended-related-work} expands this discussion through a structured narrative review with explicit search scope, screening criteria, and comparison dimensions.

\paragraph{Compositional representation and grounding.}
\citet{mohr2026concepts} study how encoder representations support different forms of concept composition. We build on this, extending schema retrieval with a hierarchy-guided ranking score that combines embedding similarity with schema structure, integrating this directly into query construction. \citet{carvalho2025montague} provide semantic distinctions for modifier consistency; we use these to motivate the representation and explicitly flag the diagnostic benchmark's subsective stratum as an approximate operationalization, not a strict instantiation, of that semantic category.

\paragraph{Structured semantic parsing.}
\citet{banerjee2023gettqa} use graph embeddings in text-to-text KG question answering. \citet{saparina2021sparqling} use intermediate question decompositions to separate interpretation from query realization. Our named process stages similarly separate responsibilities, while leaving their internal representation to the existing agentic stages. \citet{grvkbqa2025} distinguish logical-structure generation, grounding, and validation.

\paragraph{Schema-guided agentic construction.}
SPINACH \citep{liu2024spinach} grounds queries through interactive information navigation. SAGA \citep{saga2026} uses type information during property exploration and partial-query construction, while SchemaForge \citep{schemaforge2025} combines schema realization with counterfactual answer checks. Our retrieval intervention combines a composition-aware encoder with a type-compatibility feature and hierarchy depth before candidate selection. The optional solver checks extensions within a restricted candidate search. \citet{mkgqagent2025} provide a further agentic comparison across languages.

\paragraph{Reflective refinement and self-correction.}
\label{sec:relwork-llm}
\citet{mohr2026reflective} direct SQL diagnostics to stage-level prompt updates. We adapt this design to grounding, binding, and SPARQL construction, and make requirement accumulation explicit through scoped specifications and conditional preservation results. QueryAgent \citep{huang2024queryagent} uses environmental feedback for self-correction; InteracSPARQL \citep{jian2026interacsparql} supports query refinement through natural-language explanations. Our formal account separately specifies conditions for preserving validated interpretations as the staged process evolves at inference time (Section~\ref{sec:inference-time-learning}). The present experiments establish the effect of stage-attributed prompt rewriting, while certification of accumulated requirements remains a separate evaluation obligation.

\paragraph{Evaluation settings.}
Demonstration selection and prompting strongly affect text-to-SPARQL results \citep{dabramo2025llmtext2sparql}. Raw-question evaluation must therefore be distinguished from settings that supply gold entities or other annotations \citep{Banerjee_2022,diallo2023evaluation}. We keep the retrieval ablations within a fixed generation configuration and report external systems under their own metric names and conditions. These comparisons provide context rather than controlled estimates of cross-system superiority.

\section{Conclusion}

We have formalised NL2SPARQL as a prompt-parameterised grounding, binding, and construction process, with a composition-aware encoder inside grounding and reflective feedback around the complete process. Scoped requirement accumulation provides a preservation discipline in which an admitted specification retains prior requirements while restricting admissible stage behaviours, and a second result establishes retention of learned interpretation mappings under explicit compatibility contracts. These are conditional results whose premises the reported experiments do not establish (Limitations).

These retrieval and refinement gains hold consistently across targets, backbones, and composition-related strata, with diagnostics also identifying where the restricted solver loses accuracy (Section~\ref{sec:baseline-results}). Implementing and evaluating the structured requirement protocol is the next step needed to connect the formal preservation discipline to observed end-to-end behaviour.

\section*{Limitations}

The monotone requirement formulation specifies an admission contract requiring a declared requirement logic, consistency checks, and a renderer-preservation procedure, none of which the reported prompt-rewriting runs supply; a consistent requirement can still be mistaken, and a model can fail to obey it. Theorem~\ref{thm:conservative-learning}'s retention guarantee holds for fixed, validated contexts under identifiability assumptions requiring independent certificates; persistent cross-question learning, cross-session retention, and monotone answer accuracy remain empirically unestablished.

Hierarchical ranking depends on type resolution, schema coverage, candidate-pool recall, relation direction, and the chosen depth weight; its class rule can suppress finer concepts when the resolved type is too broad, and it supplies a domain-anchor feature rather than full domain/range inference. The restricted solver's strict verification is similarly vulnerable to KG snapshot incompleteness, rejecting a correct binding whose supporting triple is simply missing. This is an inherent trade-off of the verified design (Appendix~\ref{app:solver-case-analysis}).

Backbone comparisons cover only two targets, and the repeat runs, on two small benchmarks, do not establish a noise threshold for every ablation; bootstrap intervals from a single run describe question-sampling uncertainty rather than run-to-run variation. External comparator results may differ in endpoint versions and evaluation settings, and the recorded QALD-9-Plus DBpedia and Spider4SPARQL targets have not been reconciled with current repository split sizes (Appendix~\ref{app:selection-criteria}).

\section*{Ethical Considerations}

This work operates over public, general-domain knowledge graphs (DBpedia, Wikidata) and does not involve collection of personal or sensitive data. Any deployed use of the resulting orchestrated system should retain execution validation and, in higher-stakes settings, human review rather than trusting grounded answers by default. Benchmark artifacts (QALD-9-Plus, QALD-10, LC-QuAD 2.0, Spider4SPARQL) carry their own licenses and provenance; we use them for evaluation only and would check redistribution terms before releasing any derived artifact.

\section*{Acknowledgments}

This work was supported by Merck KGaA, Vevey, Switzerland, as part of an academic--industrial collaboration with the Idiap Research Institute.

\bibliography{references,references_literature,references_models,references_selection}

\clearpage
\appendix
\raggedbottom
\section*{Appendix Guide}

The appendices follow the progression from method and guarantees to evaluation
and related work.
\begin{description}[leftmargin=0em,labelwidth=0em,labelsep=0em,style=nextline]
  \item[Method and guarantees.]
  Appendix~\ref{app:approach-details} gives the staged interfaces and worked example.
  Appendix~\ref{app:formalisms} presents the algorithms and component properties.
  Appendix~\ref{app:conservative-learning} develops identification, preservation
  contracts, and the proof of Theorem~\ref{thm:conservative-learning}.
  \item[Evaluation design.]
  Appendix~\ref{app:selection-criteria} covers benchmark and baseline selection,
  scored coverage, and diagnostic-set construction.
  Appendix~\ref{app:base-models} documents the language models.
  Appendix~\ref{app:transition} specifies configurations, uncertainty estimates,
  and conditioned diagnostics.
  \item[Results and literature.]
  Appendix~\ref{app:additional-results} reports the supporting diagnostics, the refinement-budget
  sweep, supplementary comparisons, and the solver-loss bound and case-level attribution.
  Appendix~\ref{app:extended-related-work} expands the literature review.
\end{description}

\section{Staged Query Process and Worked Example}
\label{app:approach-details}

This appendix supplies the interfaces, ranking definitions, and diagnostic update equations used in Section~\ref{sec:method}, followed by the mapping difficulties and worked example introduced in the main text. Appendix~\ref{app:formalisms} establishes the corresponding conditional properties of retrieval, construction, and refinement.

\subsection{Accumulated States}

Let $\mathcal A=(\mathrm{ground},\mathrm{bind},\mathrm{construct})$ be the ordered stage set and $\mathbf p=(p_a)_{a\in\mathcal A}$ its prompt family. Starting from $x_0=(q,\KG,\Schema)$, the stages in Equation~\ref{eq:pi-composition} act as
\begin{align*}
x_g&=\pi_{\mathrm{ground}}^{p_{\mathrm{ground}},E_\theta}(x_0),\\
x_b&=\pi_{\mathrm{bind}}^{p_{\mathrm{bind}}}(x_g),\\
Q&=\pi_{\mathrm{construct}}^{p_{\mathrm{construct}}}(x_b).
\end{align*}
The grounding state retains the original question and context, composed spans and their roles, ranked class and property alternatives, and provisional class or entity anchors. The binding state adds associations between those interpretations and entity references, literals, and filter conditions. Construction selects bindings, establishes shared variables and operator scope, and returns serialized SPARQL. Execution produces $\operatorname{Exec}_{\KG}(Q)$, including answers or an execution diagnostic. These interfaces describe functional responsibilities without assuming that the evaluated software materializes an additional query graph or semantic representation.

\subsection{Composition and Hierarchical Ranking}

The taxonomy of \citet{carvalho2025montague} and \citet{mohr2026concepts} distinguishes intersective modification, which supplies an independent conjunct, from subsective modification, which restricts the head in a head-dependent way. Relational composition introduces another argument. Modal and privative expressions qualify membership or describe a non-genuine instance of the apparent head concept. The taxonomy motivates representation and stratifies evaluation. The diagnostic benchmark's subsective stratum requires a narrower interpretation than strict subsective semantics, as explained in Appendix~\ref{app:diagnostic-construction}.

Writing $s=\omega_\tau(h,m)$ for an expression composed from head $h$ and modifier or relational argument $m$, the representation is
\[
z_s=E_\theta(\omega_\tau(h,m)).
\]
The notation $\omega_\tau$ describes linguistic composition. It introduces neither an additional encoder argument nor a separately learned composition operator. The fixed encoder $E_\theta=\Enc_\theta$ receives the complete expression.

For a nearest-neighbour pool $\mathcal E_s\subseteq\Schema$, the retained list is
\begin{equation}
\mathcal C(s)=\operatorname{TopK}_{e\in\mathcal E_s}c(s,e).
\label{eq:candidates}
\end{equation}
The structural factor in Equation~\ref{eq:c-score} uses a candidate anchor,
\begin{equation}
\operatorname{Anchor}(e)=
\begin{cases}
e,&e\text{ is a class},\\
\operatorname{dom}(e),&e\text{ is a property}.
\end{cases}
\label{eq:anchor}
\end{equation}
For resolved $t_s$ and $a_e=\operatorname{Anchor}(e)$, it is
\begin{equation}
\kappa_\Sigma(t_s,e)=
\begin{cases}
w(d_\Sigma(t_s,a_e)),&t_s\sqsubseteq a_e,\\
0,&\text{otherwise},
\end{cases}
\label{eq:kappa}
\end{equation}
where $w\geq0$ is non-increasing with ancestor distance. A property's contextual type refers to its proposed subject, making the direction of its relation relevant. For classes, the rule admits the resolved class and its ancestors. A broad contextual type may therefore fail to retain a more specific class candidate.

The unresolved-type branch leaves similarity unchanged. Missing property-domain metadata and ambiguous type resolutions need an explicit backend policy. Equation~\ref{eq:kappa} supplies no compatibility guarantee for those branches. The score combines a binary compatibility feature with a continuous hierarchy weight and ranks individual candidates. It does not jointly reconcile all query elements. Since zero-scored candidates can remain in a top-$k$ list, scoring alone also does not enforce hard admissibility. Appendix~\ref{app:formalisms-retrieval} states the assumptions needed for the claimed ranking properties.

\subsection{Construction Scope}

The symbolic mode keeps the upstream anchor and shape fixed. It searches only the supplied property candidates and forbids reuse within the slot assignment. Live checks test whether an extension is supported by the endpoint under the implemented predicate. They do not establish that the interpretation expresses the question's intended meaning. Appendix~\ref{app:formalisms-solver} specifies the supported shapes, backtracking procedure, verification predicate, and conditional guarantees. These restrictions explain why fallback generation is part of the construction interface.

\subsection{Scoped Feedback and Replay}

A requirement $\delta=(g_\delta,\varphi_\delta,\nu_\delta)$ consists of an applicability condition, a behavioural obligation, and diagnostic provenance. Its meaning is the scoped implication $g_\delta\Rightarrow\varphi_\delta$. A construction requirement can, for example, preserve a modifier's attachment to the employer variable when the interpreted phrase places that modifier on the employer. The guard limits its application to that interpretation and relevant schema context.

For the current query, execution result $\mathcal O^{(r)}=\operatorname{Exec}_{\KG}(Q^{(r)})$, and construction trace $\mathcal T^{(r)}$, reflection computes
\begin{align}
d^{(r)}&=\operatorname{Diagnose}(q,Q^{(r)},
\mathcal O^{(r)},\mathcal T^{(r)}),\nonumber\\
(a^\star,\Delta\Gamma^{(r)})&=
\operatorname{Reflect}(d^{(r)},\mathcal T^{(r)},\Gamma^{(r)}).
\label{eq:reflect}
\end{align}
Admission checks consistency of the enlarged specification and preservation of its requirements in the rendered prompt. Conflicting or uncertified additions leave the prompt unchanged. Accepted additions follow Equation~\ref{eq:requirement-update}, preserving the previous requirement set at every stage.

Replay retains the prefix before the attributed stage and recomputes that stage and its successors. Diagnostics are recomputed for the latest query, including the final query produced within the refinement budget. Interpreter checks and LLM judgements supply diagnostic evidence without gold SPARQL. Acceptance is relative to those checks. Budget exhaustion, a failed or inconclusive admission check, and the absence of a new requirement have distinct stopping statuses. Appendix~\ref{app:formalisms-refinement} defines those statuses and separates specification monotonicity, rendering fidelity, model compliance, and answer accuracy.

\subsection{Dimensions of Mapping Difficulty}
\label{app:mapping-difficulty}

Schema grounding is hard for more than one reason at once, and this paper's mechanism and evaluation address some of those reasons directly, some only partially, and some not at all. This appendix names the dimensions explicitly, so the paper's own scope is stated rather than implied, and works through one concrete example that exercises the two dimensions the paper actually tests.

\label{app:mapping-difficulty-dims}

Three dimensions connect directly to a result reported in this paper:

\begin{description}[style=nextline,leftmargin=0em,labelwidth=0em,labelsep=0em]
  \item[Semantic distance.] A question expression and the schema element it should resolve to can differ arbitrarily in wording, granularity, or definition, with no lexical overlap to exploit. This is the failure mode $\Enc_\theta$ targets, and RH1 tests directly (Section~\ref{sec:empirical}).
  \item[Compositional dependence, in the narrow sense this paper tests.] How a modifier composes with its head concept (intersectively, subsectively, or relationally, Section~\ref{sec:composition}) constrains which schema elements the enclosing expression can resolve to. This is what motivates the composition taxonomy. The supporting analysis measures retrieval gains across operational benchmark strata, including an approximately-operationalized subsective stratum. It is \emph{not} the broader claim that a full query's several schema elements are jointly selected and reconciled against each other under shared constraints during retrieval: The grounding stage retrieves candidates for each mention independently, and connecting them (shared variables, multi-hop joins, cross-element consistency) is left to the query-synthesis stage, either its free-text generation or, for two narrow query shapes, the restricted symbolic solver's backtracking search over property-to-slot bindings (Section~\ref{sec:solver}) -- which itself only searches over properties given an already-fixed class or entity anchor, not over the anchor choice itself, and only for single-hop shapes. Compositional dependence in the full, multi-element-and-anchor sense is a real source of difficulty this work does not claim to solve.
  \item[Schema scale.] Retrieval has to remain accurate as the candidate pool grows. DBpedia and Wikidata expose large graph-wide schema vocabularies, while Spider4SPARQL uses smaller per-domain schemas. The paper does not run a dedicated schema-scale ablation, but the reported recall@10 comparison (Figure~\ref{fig:prospective-recall-curves}) and the spread across benchmarks targeting differently sized schemas (Spider4SPARQL's small, per-KG schemas versus DBpedia's and Wikidata's schema-wide ones) speak to this dimension indirectly.
\end{description}

Two further dimensions are real but outside this paper's scope, named here rather than left implicit:

\begin{description}[style=nextline,leftmargin=0em,labelwidth=0em,labelsep=0em]
  \item[Representation multiplicity.] The same underlying concept can be realized through different properties, class restrictions, paths, or qualified statements in the same KG. Nothing in this paper's evaluation tests robustness to multiple valid, equivalent schema realizations of one concept.
  \item[Context and scope.] Temporal, comparative, negative, or role-dependent meaning has to attach to the correct relation or statement. The subsective stratum samples a narrow form of temporal or status-dependent wording as a proxy for head-dependent modification, but does not separately validate general temporal, comparative, or negative scope.
\end{description}

\subsection{A Worked Example}
\label{app:mapping-difficulty-example}

The following illustrative example (constructed for exposition, not drawn from any benchmark) exercises semantic distance and compositional dependence together, the two dimensions the paper's mechanism and taxonomy directly address. Consider a small investment schema with classes \texttt{ex:Organisation}, \texttt{ex:Company}, and \texttt{ex:PrivateEquity} (a holding type restricted to unlisted companies), and properties \texttt{ex:holds}, \texttt{ex:issuer}, and \texttt{ex:hqCountry}. Consider the question:

\begin{quote}
Which organisations own stakes in unlisted companies headquartered in the same country as themselves?
\end{quote}

``Own stakes in unlisted companies'' shares no lexical material with \texttt{ex:PrivateEquity}, the class it should resolve to (semantic distance); a public-equity or debt holding would fail the ``unlisted'' condition despite lexical closeness to ``stakes.'' Resolving this one class is not sufficient to answer the question, though: the query also has to bind the organisation's own \texttt{ex:hqCountry} and the held company's \texttt{ex:hqCountry} to the same variable (compositional dependence, in the query-construction sense this paper's mechanism does not claim to solve directly). Figure~\ref{fig:mapping-difficulty-example} gives the complete query.

\begin{figure}[t]
\begin{lstlisting}[basicstyle=\ttfamily\scriptsize, breaklines=true, columns=fullflexible]
PREFIX ex: <https://example.org/>
SELECT DISTINCT ?org WHERE {
  ?org a ex:Organisation ;
       ex:holds ?holding ;
       ex:hqCountry ?country .
  ?holding a ex:PrivateEquity ;
           ex:issuer ?company .
  ?company a ex:Company ;
           ex:hqCountry ?country .
}
\end{lstlisting}
\caption{\textbf{Shared variables preserve the same-country constraint.} The shared \texttt{?country} binding enforces the same-country condition across the organisation and the company it holds a stake in. All identifiers are illustrative.}
\label{fig:mapping-difficulty-example}
\end{figure}

\FloatBarrier
\section{Algorithms and Component Guarantees}
\label{app:formalisms}

The results below concern the score and the restricted search actually defined in Section~\ref{sec:method}. Neither ranking nor endpoint support supplies a semantic-correctness theorem for the complete NL2SPARQL mapping.

\subsection{Hierarchical Candidate Ranking}
\label{app:formalisms-retrieval}

Fix a span $s$, a finite candidate pool $\mathcal E_s$, and a resolved contextual type $t_s$. Let $b_s(e)=\operatorname{sim}(E_\theta(s),E_\theta(e))$. For candidates with available anchors, Equations~\ref{eq:anchor}--\ref{eq:kappa} define the ranked score. The anchor test is directional: for properties it compares the proposed subject type with the property's domain; for classes it admits the resolved class and its ancestors. It is not a symmetric test of whether two class descriptions overlap.

\begin{algorithm*}[t]
\normalsize
\caption{\textbf{Rank candidates using composition and schema compatibility.}}
\label{alg:retrieval}
\begin{lstlisting}[style=paperalgorithm]
Input: span $s$, encoder $E_\theta$, pool $\mathcal E_s$, type $t_s$, budget $k$
$z_s \gets E_\theta(s)$
for $e\in\mathcal E_s$:
    $b_s(e)\gets\operatorname{sim}(z_s,E_\theta(e))$
    if type and anchor are resolved:
        $c(s,e)\gets b_s(e)\kappa_\Sigma(t_s,e)$
    else:
        $c(s,e)\gets b_s(e)$
return top-$k$ candidates under $c(s,\cdot)$
\end{lstlisting}
\end{algorithm*}

Algorithm~\ref{alg:retrieval} makes the neutral unresolved case explicit as an interface convention. The supplied experimental summaries do not identify all backend decisions for absent domain metadata, multiple domains, ambiguous types, or the numerical depth-weight function. Their reproduction requires those configuration details; the following propositions apply to the resolved branch only.

\begin{proposition}[Suppression by the anchor test]
\label{prop:floor}
If $t_s\not\sqsubseteq\operatorname{Anchor}(e)$, then $c(s,e)=0$. If $b_s(e')\geq0$ and $w\geq0$ for every candidate $e'$ that passes the test, a failed-test candidate cannot have a strictly larger score than a passed-test candidate.
\end{proposition}
\begin{proof}
The zero branch in Equation~\ref{eq:kappa} applies to $e$. Every passed-test score is a product of two nonnegative factors.
\end{proof}

Zero is a score, not removal from the candidate pool. Failed-test candidates may survive top-$k$ truncation or ties. Nearest-neighbour retrieval alone does not imply nonnegative cosine similarities, so the nonnegativity assumption is needed for the ordering statement. Failure to establish subclass entailment can also reflect incomplete metadata rather than genuine type incompatibility.

\begin{proposition}[Hierarchy ordering at equal similarity]
\label{prop:depth}
Suppose two candidates pass the anchor test, have the same nonnegative similarity $b$, and have ancestor distances $d_1\leq d_2$. If $w$ is non-increasing, then $c(s,e_1)\geq c(s,e_2)$.
\end{proposition}
\begin{proof}
$c(s,e_j)=b\,w(d_j)$ and $b\geq0$ preserve the ordering of the weights.
\end{proof}

These statements specify how the ranking changes. They neither ensure that the correct candidate is in $\mathcal E_s$ nor establish a recall improvement, which remains an empirical question.

\begin{figure*}[t]
\centering
\resizebox{\textwidth}{!}{%
\input{figures/generated/fig_mechanism.include.tex}\unskip%
}
\caption{\textbf{Type compatibility and hierarchy depth reshape candidate scores.} Two independent toy applications of Equation~\ref{eq:c-score}.
All similarity values are invented for illustration, not measured retrieval results.
The gate is $\mathbb{1}[t_s\sqsubseteq\mathrm{Anchor}(e)]$.
(a)~The subject type is checked against the property's declared domain, in the
 displayed predicate direction. (b)~The resolved type is checked against each candidate
class: \texttt{Novelist} is compatible with itself and its ancestor \texttt{Person},
with hierarchy distance reducing the latter's score. An unmet subclass condition
sets the score to zero; this alone does not remove the candidate from a top-$k$ list
or establish logical disjointness. Type resolution is an input to the mechanism.}
\label{fig:mechanism}
\end{figure*}

\subsection{Restricted Graph-Supported Construction}
\label{app:formalisms-solver}

Fix an upstream anchor $a$, a supported query shape, its bound values, and finite property lists $\mathcal C_1,\ldots,\mathcal C_m$, with $m\geq1$. A partial assignment $\beta$ maps initial slots to distinct properties. Let $V_{\KG}(a,\beta)$ be the implemented endpoint predicate for the cumulative partial pattern, evaluated with all constraints whose arguments are already bound. This predicate provides graph support. It is not a test of equivalence to the natural-language question.

\begin{algorithm*}[t]
\normalsize
\caption{\textbf{Search supported assignments while preserving branch state.}}
\label{alg:solver}
\begin{lstlisting}[style=paperalgorithm]
Input: anchor $a$, lists $\mathcal C_1,\ldots,\mathcal C_m$, graph $\KG$
Search($i,\beta$):
    if $i>m$: return $\beta$
    for $e\in\mathcal C_i\setminus\operatorname{im}(\beta)$:
        $\beta'\gets\beta\cup\{i\mapsto e\}$   # fresh copy
        if $V_{\KG}(a,\beta')=\mathrm{pass}$:
            $\gamma\gets\operatorname{Search}(i+1,\beta')$
            if $\gamma\neq\mathrm{Fail}$: return $\gamma$
    return Fail
$\beta\gets\operatorname{Search}(1,\emptyset)$
if $\beta\neq\mathrm{Fail}$: return $\operatorname{Serialize}(a,\beta)$
else: return default generator output
\end{lstlisting}
\end{algorithm*}

\begin{proposition}[Verification of returned assignments]
\label{prop:soundness}
Every assignment returned by \textnormal{\textsc{Search}} has a complete branch for which each cumulative prefix passed $V_{\KG}$ at the time it was extended. In particular, its final assignment passed that predicate.
\end{proposition}
\begin{proof}
A recursive call receives only the fresh assignment inside the guarded branch. At depth $m+1$, every slot has been extended through such a branch. Copying $\beta$ prevents assignments from failed alternatives from entering another branch. Template serialization preserves the selected assignment.
\end{proof}

\begin{proposition}[Finite candidate-relative search]
\label{prop:solver-completeness}
Assume a fixed endpoint snapshot, terminating exact checks, exhaustive traversal without a resource cutoff, and deterministic serialization. The search terminates and finds an assignment whenever a distinct-property assignment in $\mathcal C_1\times\cdots\times\mathcal C_m$ passes $V_{\KG}$ at every prefix.
\end{proposition}
\begin{proof}
The search tree has finite depth $m$ and finite branching. Every eligible child whose prefix passes the predicate is visited unless an earlier successful branch returns. Induction on the remaining slots establishes discovery of any all-pass branch.
\end{proof}

The result is relative to the supplied candidates, fixed anchor, supported shape, property-exclusion rule, and prefix predicate. A complete pattern can be missed if a premature filter rejects one of its incomplete prefixes. Unknown checks, timeouts, and endpoint changes invalidate the corresponding completeness assumptions. Solver failure is therefore not evidence of question unanswerability. A semantically wrong property may pass every graph check, and default-generator outputs do not inherit the solver's verification property.

\subsection{Scoped Requirements and Targeted Refinement}
\label{app:formalisms-refinement}

This appendix specifies a structured version of the refinement mechanism in
Section~\ref{sec:refinement}. It distinguishes accumulated requirements from
their rendering as prompts and states exactly what is monotonic. The existing
experiments evaluate stage-attributed prompt rewriting, adapted from
\citet{mohr2026reflective}; they do not empirically certify the structured
admission procedure specified here. The admission procedure therefore has a distinct implementation and evaluation obligation.

\paragraph{Typed stages and accumulated state.}
Let
\[
\begin{gathered}
\mathcal A=\{\mathrm{ground},\mathrm{bind},\mathrm{construct}\},\\
\mathrm{ground}\prec\mathrm{bind}\prec\mathrm{construct}.
\end{gathered}
\]
Writing $x_0=(q,\KG,\Schema)$, the accumulated states of
Section~\ref{sec:mappings} have the form
\[
\begin{aligned}
x_g&=(x_0,C),\\
x_b&=(x_g,V),\\
x_c&=(x_b,Q).
\end{aligned}
\]
Here $C$ records spans, modifier roles, candidate schema interpretations,
and provisional anchors; $V$ records entity, literal, and filter bindings;
and $Q$ is the constructed SPARQL query. The construction mapping returns
$Q$; $x_c$ retains its input alongside that output for diagnosis.
Their types may retain unresolved alternatives. At stage $a$, write
$x$ for the incoming state and $u\in\mathcal U_a(x)$ for a well-typed
addition, so the outgoing state is $(x,u)$.
These extensions describe the diagnostic trace representation. Grounding and binding return their accumulated states as in Equation~\ref{eq:pi-composition}; construction returns $Q$ and is lifted to $(x_b,Q)$ only for the trace. A stage may instead report failure. These are functional interfaces, not a claim that every
intermediate state is separately materialized in the evaluated software.
Generation can be stochastic; functional notation denotes a realised execution with fixed backend and decoding settings. The encoder $\Enc_\theta$ of Section~\ref{sec:encoder} remains fixed inside
$\pi_{\mathrm{ground}}^r$ throughout refinement.

\begin{definition}[Scoped requirement]
A requirement for stage $a$ is a triple
$\delta=(g_\delta,\varphi_\delta,\nu_\delta)$, where $g_\delta(x)$ is an applicability guard,
$\varphi_\delta(x,u)$ is a predicate on a typed stage extension, and $\nu_\delta$ records
provenance, including the originating diagnostic and scope. Its meaning is
\[
\operatorname{req}_\delta(x,u)
\;:=\;g_\delta(x)\Rightarrow\varphi_\delta(x,u).
\]
For a finite requirement set $\Gamma_a^r$, define
$R_{\Gamma_a^r}=\bigwedge_{\delta\in\Gamma_a^r}\operatorname{req}_\delta$.
The empty conjunction is true. Provenance supports inspection; it does not
by itself establish that a requirement is correct.
\end{definition}

Guards must refer to information available at the named stage. For example,
a rule about retaining a country-to-capital relation can be scoped to a
particular relational interpretation and schema version. A diagnostic on
one question does not justify applying its rule to every question.

\paragraph{Specification and prompt rendering.}
Let $B_a$ be a fixed base prompt with declared base requirements $b_a(x,u)$.
The structured specification and its rendered prompt are
\[
S_a^r=b_a\wedge R_{\Gamma_a^r},
\qquad
p_a^r=\operatorname{Render}(B_a,\Gamma_a^r).
\]
Requirement accumulation is semantic: the clauses denote predicates, not
merely strings appended to a prompt. Rendering must preserve the complete
current specification. This is a contract on the renderer, not a guarantee
that an LLM obeys the rendered prompt. In particular, rewriting free text
without retaining an explicit $\Gamma_a^r$ does not establish this contract.

For each stage, fix a declared input domain $\mathcal D_a$ and define its
specification-admissible extensions by
\[
\mathcal F_a^r(x)=
\{u\in\mathcal U_a(x):S_a^r(x,u)\},
\qquad x\in\mathcal D_a.
\]
Scopes, schema versions, and predicate interpretations are fixed when
comparing revisions. If these change, the comparison requires a new
specification rather than an implicit alteration of an existing clause.

\begin{definition}[Prompt strengthening]
For prompts with declared specifications, write
\[
p_a\preceq p'_a\quad\Longleftrightarrow\quad
\operatorname{Req}(p'_a)\models_{\mathcal D_a}\operatorname{Req}(p_a).
\]
Entailment makes this relation a preorder. For a faithful renderer,
$\operatorname{Req}(\operatorname{Render}(B_a,\Gamma_a))$
is equivalent to $b_a\wedge R_{\Gamma_a}$. The order concerns these
declared semantics and assumes they have a fixed interpretation.
\end{definition}

\begin{definition}[Admission]
For the attributed stage $a^\star$, a proposed addition
$\Delta\Gamma$ is admitted only if its predicates are well typed, its
scopes and provenance are explicit, and the proposed specification
\[
S'=b_{a^\star}\wedge
R_{\Gamma_{a^\star}^r\cup\Delta\Gamma}
\]
passes requirement preservation and consistency checks:
\[
\begin{gathered}
S'\models_{\mathcal D_{a^\star}}S_{a^\star}^r,\\
\forall x\in\mathcal D_{a^\star}:\quad\mathcal F'(x)\ne\emptyset,
\end{gathered}
\]
where $\mathcal F'$ is the feasible-extension set defined by $S'$.
Admission also requires a renderer satisfying the specification-preservation
contract. An inconclusive check does not count as admission.
\end{definition}

The entailment check concerns the meaning of the accumulated requirements.
For exact conjunction with unchanged clauses it follows directly, but it
must remain valid under any normalization of those clauses. Consistency is
per input: a single satisfying pair outside a rule's guard would not rule
out a contradiction inside its scope. These checks require a specified
requirement logic and an appropriate decision procedure or certificate;
no general procedure for arbitrary natural-language predicates is assumed.
A check restricted to the current input establishes only that local result.

\begin{proposition}[Monotonic specification]
Suppose admitted updates use
$\Gamma_{a^\star}^{r+1}=\Gamma_{a^\star}^r\cup\Delta\Gamma$ and preserve
all other requirement sets. Then, for every stage and every
$x\in\mathcal D_a$,
\[
\Gamma_a^r\subseteq\Gamma_a^{r+1},
\qquad
\mathcal F_a^{r+1}(x)\subseteq\mathcal F_a^r(x).
\]
At the updated stage, admission additionally ensures that the new feasible
set is nonempty on the declared domain.
\end{proposition}
\begin{proof}
Set union preserves each earlier clause. A conjunction containing all
earlier clauses entails their conjunction, so every newly admissible
extension was previously admissible. The consistency check gives
nonemptiness. Unchanged stages have equal specifications. A faithful renderer therefore gives $p_a^r\preceq p_a^{r+1}$.
\end{proof}

If a generated extension complies with its current specification, it
therefore satisfies every earlier requirement. This conditional statement
does not establish compliance, nested model-output distributions, improved
answer accuracy, or semantic equivalence to the question. In particular,
two compliant prompts can induce different samples from their common
admissible set.

\begin{proposition}[Targeted update locality]
Assume $B_a$ and the renderer are fixed and only
$\Gamma_{a^\star}$ changes. Then $p_a^{r+1}=p_a^r$ for
$a\ne a^\star$. Retaining the prefix state immediately before
$a^\star$ and recomputing $a^\star$ and every later stage produces a
trace in which each replayed stage consumes its newly computed predecessor.
\end{proposition}
\begin{proof}
Unchanged renderer inputs give unchanged prompts. Recomputing the suffix
in stage order supplies each stage with the preceding stage's new output.
The retained prefix is unchanged by construction.
\end{proof}

Locality concerns the updated policy and retained prefix, not unchanged
downstream outputs. In particular, changing grounding can alter bindings
and construction even when their prompts remain fixed.

\begin{algorithm*}[t]
\normalsize
\caption{\textbf{Admit scoped requirements before replaying the affected stages.}}
\label{alg:refinement}
\begin{lstlisting}[style=paperalgorithm]
Input: $q,\KG,\Sigma,B,\Gamma$ and revision budget $T$
Require: initial specifications and renderer
         satisfy the admission contracts
$r\gets0$; $p_a\gets\operatorname{Render}(B_a,\Gamma_a)$
$X\gets\operatorname{Run}(\mathrm{ground},\mathrm{bind},\mathrm{construct};p)$
while true:
    $d\gets\operatorname{Diagnose}(q,X,\Gamma,\KG)$
    # execute and inspect the latest query
    if all checks pass: return $(X,\mathrm{Accepted},d)$
    if $r=T$: return $(X,\mathrm{Budget},d)$
    $(a^\star,\Delta\Gamma)\gets\operatorname{AttributeAndPropose}(d,X,\Gamma)$
    if no new requirement: return $(X,\mathrm{NoAddition},d)$
    if admission is not certified:
        return $(X,\mathrm{Conflict},d)$
    $\Gamma_{a^\star}\gets\Gamma_{a^\star}\cup\Delta\Gamma$
    $p_{a^\star}\gets\operatorname{Render}(B_{a^\star},\Gamma_{a^\star})$
    $X\gets\operatorname{ReplaySuffix}(X,a^\star,p)$
    $r\gets r+1$
\end{lstlisting}
\end{algorithm*}

\paragraph{Stopping and scope of the claims.}
Algorithm~\ref{alg:refinement} evaluates diagnostics afresh after every
replay, including the last permitted revision. Diagnostics combine current
stage-requirement checks, syntax and execution checks, and the declared
semantic checks; failure to construct a query is itself diagnostic evidence.
\textsc{Accepted} means that these checks pass.
\textsc{Budget} retains a current failure after $T$ admitted updates.
\textsc{Conflict} records failed or inconclusive admission, without applying
the proposed update. \textsc{NoAddition} records an empty or semantically
redundant proposal. None certifies natural-language correctness. The finite
budget ensures termination if individual calls terminate; it supplies no
convergence theorem. In particular, a consistent but mistaken requirement
can exclude the intended query. The measured benefit of prompt rewriting
therefore remains distinct from the conditional specification properties
proved above.

\FloatBarrier
\section{Conservative Inference-Time Learning}
\label{app:conservative-learning}

This appendix specifies the conditions of Theorem~\ref{thm:conservative-learning}. Its contracts are sufficient conditions for accepting an update. They are not consequences of endpoint success, requirement accumulation, or the reported experiments.

\subsection{Contextual Mappings and Protected History}

A learned mapping $\mu$ is a partial function from an expression and its context to an interpretation. Context includes the relevant argument roles and schema information. The bank $M_t$ stores versioned mappings, applicability guards, and provenance. Its validated entries have a declared semantic scope. Other entries remain hypotheses that may be revised. The procedure applying the bank is the composition of named stages,
\[
\pi_t=
\pi_{\mathrm{construct},t}\circ
\pi_{\mathrm{bind},t}\circ
\pi_{\mathrm{ground},t}^{E_\theta}.
\]
Each stage uses its named prompt $p_{a,t}$. Refinement may change these prompts, parsing, dispatch, and the composition of interpretations, while $E_\theta$ remains fixed. Revisions preserve the order of the three named interfaces. A different stage topology requires an additional certified alignment to those interfaces.

A protected record is $h=(\xi_h,I_h,Q_h,\tau_h)$, where $\xi_h$ contains the question and recorded context, $I_h$ is its validated interpretation, $Q_h$ is the resulting query, and $\tau_h$ records the stage states, mapping invocations, and originating process version $v(h)\leq t$. The finite set of records is $H_t$ and may span several earlier versions. Validation is assumed sound for the declared interpretation semantics. It is not identified with obtaining nonempty answers.

Write $\simeq_c$ for semantic equivalence in context $c$. The semantics must specify argument roles, operator scope, and the applicable query semantics, including multiplicities where relevant. Graph and schema versions are held fixed during comparisons, unless a separate transport certificate relates versions. If equivalence is defined only over one snapshot, the resulting claim has exactly that restricted scope. Interpreting an expression correctly still requires evidence connecting its language to the semantic object.

Protected mapping uses form finite sets $U_{\mu,H_t}$ extracted from the traces. A certificate over these uses does not establish correctness over every input satisfying the mapping's guard. Extending protection to a domain requires an additional domain-level certificate.

\subsection{Identification and Edit Localization}

Let $\mathcal V(D)$ be a candidate class of contextual mappings defined throughout a declared scope $D$. Assume the intended mapping $\mu^*$ is realizable in this class. For observations $\mathcal E_t$ and a declared consistency predicate $\operatorname{Pass}$, retain
\[
\mathcal V_t(D)=\{\mu\in\mathcal V(D):
\forall e\in\mathcal E_t\;\operatorname{Pass}(e,\mu)\}.
\]
Evidence is sound if every observation passes $\mu^*$, and separating if every inequivalent candidate fails at least one observation. These conditions give semantic identification,
\[
\mathcal V_t(D)/{\simeq_D}=\{[\mu^*]_D\}.
\]
Indeed, soundness and realizability retain the displayed class, while separation removes all others. This is a conditional identification statement, not an assertion that execution diagnostics satisfy its premises.

Context separability requires the representation and guard language to distinguish cases needing inequivalent interpretations. Overlapping guards require an explicit, coherent selection policy on protected contexts. A trace-based localization record identifies the proposed mapping, responsible stage, applicability scope, and supporting diagnostic. It justifies an edit site without requiring a uniquely identifiable repair. A failure can admit several repairs, and localization alone does not certify any of them.

These conditions concern the admission of a learned interpretation. The retention result below instead relies on preservation contracts. It does not derive semantic identification from those contracts.

\subsection{Local Contracts for a Process Revision}

For a fixed protected record, write its stored trace as
\[
s_0^h\longrightarrow s_g^h\longrightarrow
s_b^h\longrightarrow s_c^h.
\]
Here $g$, $b$, and $c$ abbreviate the named stage boundaries, and $s_c^h$ retains $Q_h$ and its incoming state. The updated procedure may use different state representations. For each stage boundary, specify a relation $R_a^h$ between old and new states. The relation $R_0^h$ identifies the same protected input context. The other relations describe the interfaces that must survive the revision.

These relations must be specified through interface meaning, independently of whether the final run happens to succeed. At grounding they preserve the roles and semantic alternatives passed onward. At binding they preserve the relevant entity, literal, attachment, join, and scope obligations under any declared variable correspondence. At each relevant boundary they also preserve designated protected mapping observations under a declared correspondence, including when parsing or decomposition changes. Representation adapters belong to these interfaces. At construction the relation must imply equivalence of the final interpretations. Relations may permit additional internal data, but cannot discard a distinction on which the protected interpretation depends.

\paragraph{Mapping preservation.}
Each protected mapping remains available, or has a replacement certified equivalent on $U_{\mu,H_t}$. Any change to its input representation requires a corresponding certified adapter. This preserves the mapping's declared meaning on protected uses. It does not assert that the updated dispatcher will select it. Dispatch, rule priority, and guard overlap are part of the stage contracts below. In particular, a rule's previous absence can matter to selection, so checking only positively recorded calls misses possible shadowing.

\paragraph{Initialization and terminal compatibility.}
Updated initialization terminates successfully on $\xi_h$, has a nonempty set of permitted outputs, and every output is related to $s_0^h$ by $R_0^h$. All preprocessing is included in initialization or grounding. At the terminal interface,
\[
R_c^h(s_c^h,z_c)
\;\Longrightarrow\;
\mathcal J(z_c)\simeq_{c_h} I_h,
\]
where the old query has validated interpretation $I_h$. This is an interface contract for query construction, not a test based on the observed answers of the complete revised pipeline.

\paragraph{Local simulation and totality.}
Let $a^-$ denote the boundary before stage $a$. For every updated state $z$ reachable at that boundary from $\xi_h$, require
\begin{equation}
\begin{split}
R_{a^-}^h(s_{a^-}^h,z)\ \Longrightarrow\qquad\qquad\\
\emptyset\ne\operatorname{Out}_{a,t+1}(z)
\subseteq\{z':R_a^h(s_a^h,z')\}.
\end{split}
\label{eq:local-stage-simulation}
\end{equation}
Every permitted invocation must also terminate successfully. The output relation covers all allowed behaviours of the stage, including selection and composition decisions. A failed, divergent, or empty-output stage therefore cannot satisfy the contract vacuously.

The direction is from new behaviours to protected old interfaces. Merely retaining one old-compatible path would permit additional incorrect paths. For an unchanged deterministic stage, the obligation is congruence on related inputs. An unchanged stochastic stage must satisfy it for all covered outputs. Modified stages require a local compatibility proof or an enforcing check over the declared domain. Keeping their prompt text unchanged does not discharge either obligation.

\subsection{Proof and Extent of Retention}

\begin{proof}[Proof of Theorem~\ref{thm:conservative-learning}]
Fix $h\in H_t$ and any covered execution of the updated procedure. Initialization establishes $R_0^h$. Apply Equation~\ref{eq:local-stage-simulation} to grounding. Totality supplies a successful successor, and every allowed successor is related to $s_g^h$ by $R_g^h$. This successor is reachable, so the binding contract applies. It yields a successfully terminating successor related to $s_b^h$ by $R_b^h$. The construction contract applies to that state and yields a successful terminal state $z_c$ related to $s_c^h$ by $R_c^h$.

Terminal compatibility gives $\mathcal J(z_c)\simeq_{c_h}I_h$. Sound validation of the historical trace certifies $I_h$, proving the displayed retention claim without identifying that trace with a run of the current process. The execution and record were arbitrary, so the result holds for every covered execution on $H_t$. Protected mapping meanings remain available by the separate mapping-preservation condition. The stage induction preserves the designated mapping observations and the validated composite interpretation. No inference from unchanged mapping names to unchanged process behaviour is made.
\end{proof}

For successive questions, an admitted update first preserves $H_t$. A new record enters $H_{t+1}$ after semantic validation and certification of the current procedure's local contracts on that record. A validated sample without this coverage remains a run-scoped observation. With this admission rule, induction over accepted updates protects every accumulated record, provided each update re-establishes the contracts for that history. Hypotheses outside the protected scope can continue to change.

Specification growth $\Gamma_a^{(r)}\subseteq\Gamma_a^{(r+1)}$ supplies neither the interface relations nor their simulation proofs. Likewise, testing one stochastic run establishes only that run's behaviour. Extending the theorem to every generated output requires contract coverage and compliance for every allowed draw, or a deterministic backend whose full transitions satisfy the contracts. Neither a finite history nor these conditional arguments imply universal-domain preservation, guaranteed discovery of correct mappings, or convergence of refinement. Persistent memory, identification certificates, and cross-query retention were not evaluated in the existing experiments.

\section{Benchmark and Baseline Selection}
\label{app:selection-criteria}

This appendix makes the evaluation scope and retention criteria explicit. It is a
retrospective account of the supplied experiment records and the selected resources,
not a preregistered sampling plan or a systematic screening of every available
benchmark and system. Selecting a benchmark family, determining which questions
contribute to a metric, and retaining an experimental result are separate decisions.

\subsection{Benchmark Families and Evaluation Targets}
\label{app:benchmark-selection}
\label{app:benchmark-protocol}

\paragraph{Inclusion criteria.}
The public benchmark families meet three task requirements: English natural-language
questions paired with SPARQL queries or reference answers; an identified RDF graph
against which schema grounding and query execution can be assessed; and a released
question split that permits the evaluation target to be specified. The selected
families provide complementary settings: DBpedia and Wikidata grounding in QALD,
complex question--query pairs at larger scale in LC-QuAD 2.0, and heterogeneous
per-domain schemas in Spider4SPARQL
\citep{perevalov2022qald9plus,usbeck2024qald10,dubey2019lcquad2,kosten2023spider4sparql}.
The four public families yield five main benchmark/KG targets because QALD-9-Plus
is evaluated on two graphs. The constructed diagnostic set is a fifth family,
analysed separately by composition-related stratum rather than added to the pooled
main-target recall.

\paragraph{Scope exclusions.}
Non-English evaluation, text-only question answering, and query tasks without a
SPARQL/RDF target lie outside this evaluation's scope. Other eligible NL2SPARQL
benchmarks are not claimed to have been screened and rejected. The selected set
supports comparisons across its recorded targets; it does not constitute a random
sample of schemas, languages, or deployment settings.

\begin{table*}[t]
\centering
\footnotesize
\tablecompact
\begin{tabular}{@{}>{\raggedright\arraybackslash}p{0.17\textwidth}>{\raggedright\arraybackslash}p{0.30\textwidth}>{\raggedright\arraybackslash}p{0.47\textwidth}@{}}
\toprule
\rowcolor{TableHead}
Evaluation target & Complementary role & Recorded target, graph, and metric coverage \\
\midrule
QALD-9-Plus, DBpedia & English grounding in an established open KG & 93-question recorded target; 92 contribute to F1. Local DBpedia 2016-10 snapshot. \\
\addlinespace[3pt]
QALD-9-Plus, Wikidata & A second open KG and its different schema & 136-question target; 125--135 contribute, depending on metric. Live Wikidata via QLever. \\
\addlinespace[3pt]
QALD-10 & Further Wikidata questions and reference-state sensitivity & 394-question test target. Live QLever execution; static and live-re-executed-gold reference scores remain distinct. \\
\addlinespace[3pt]
LC-QuAD 2.0 & Complex question--query pairs at larger scale & Recorded 6{,}046-question test target, raw questions. Dataset refers to 2018 graphs; execution uses live Wikidata via QLever. \\
\addlinespace[3pt]
Spider4SPARQL & Heterogeneous local graph schemas & Resource spans 166 custom KGs across 138 domains. Recorded 932-question development target; 905--926 contribute. Local dumps for 18 of 20 bundled graphs; two excluded for recorded data-quality reasons. \\
\addlinespace[3pt]
\rowcolor{TableAccent}
Composition diagnostic & Controlled composition-related strata & 181 constructed questions on DBpedia 2016-10: 49 intersective, 29 subsective, 49 relational, 27 modal, 27 privative. \\
\bottomrule
\end{tabular}
\caption{\textbf{Evaluation targets differ in schema scope and scored coverage.} Counts describe
the supplied experiment record, not a reconstructed item-level exclusion log. The
five main benchmark/KG targets are followed by the separate diagnostic set.
Upstream-release discrepancies for QALD-9-Plus DBpedia and Spider4SPARQL are discussed
in Appendix~\ref{app:benchmark-selection}.}
\label{tab:selection-criteria}
\end{table*}

\paragraph{Recorded subsets and upstream releases.}
The counts in Table~\ref{tab:selection-criteria} retain the experiment record.
They should not all be read as complete current upstream splits. At this revision,
the QALD-9-Plus repository lists 150 DBpedia test questions, whereas the supplied run
record targets 93; the Spider4SPARQL repository lists 1{,}034 development queries,
whereas the supplied run record targets 932
\citep{qald9plus_repository,spider4sparql_repository}. These repository metadata
were checked on 5 October 2026.
The release identifiers and question-ID mappings needed to reconcile these
subset/version differences were not supplied. We preserve the observed evaluation
counts and do not infer that the differences correspond to particular filtering
rules or unsuccessful queries.

\paragraph{Question-level inclusion and evaluability.}
Within a selected target, the intended comparison is over the same questions and
reference convention for each configuration. Infrastructure or reference-data
problems concern evaluability; a model's incorrect query, difficult grounding case,
or empty predicted answer does not by itself justify removing a question. The
available record reports incomplete or metric-dependent coverage on QALD-9-Plus
and Spider4SPARQL, including two locally excluded Spider4SPARQL graphs attributed
to data-quality problems. It does not provide the excluded item IDs, graph-specific
fault descriptions, or each metric's handling of failed queries and unavailable
answers. Consequently, the record cannot establish identical effective denominators
across every configuration, certify difficulty-independent missingness, or support
a claim that every failure was assigned zero. The reported counts therefore qualify
the answer comparisons rather than define a new benchmark split.

Two restrictions are explicitly different from this missing coverage. The conditioned diagnostic uses the
intersection of entity-and-structure-correct questions under the two compared
configurations, so its conditioned F1 targets that subset alone
(Appendix~\ref{app:diagnostic-derivation}). The auxiliary Spider4SPARQL retry comparison
uses a recorded matched $n=912$ sample and is kept separate from the main $n=932$
target (Table~\ref{tab:config-ablation}). Neither subset is substituted for a
full-target result.

\paragraph{Diagnostic-set admission.}
The 181-question diagnostic set is deliberately constructed around populated graph
patterns and composition-related distinctions. Candidate gold queries are executed
against DBpedia 2016-10, and the construction record retains candidates with checked,
in-range answers. Relational cases are downsampled from 67 to 49; modal and privative
categories are screened for semantic fit. These construction filters precede the
reported system comparison and differ from excluding a question because a compared
system answers it poorly. They also make the diagnostic set a targeted probe rather
than an estimate of naturally occurring question frequencies. The original
subsective-labelled cases remain included under that label, with their
construct-validity limitation made explicit. Appendix~\ref{app:diagnostic-construction}
describes these choices; the record does not supply a complete rejected-candidate
ledger or numerical answer-range thresholds.

\paragraph{Graph version and answer reference.}
The DBpedia evaluations use the stated 2016-10 snapshot, and Spider4SPARQL uses local
per-KG dumps. The Wikidata experiments use the live QLever service identified in the
run record, whose contents need not match the dataset's original graph version.
For QALD-10, the static reference compares predictions with released reference
answers, while the live convention re-executes the gold query on the evaluation
endpoint. These are separate reference conditions, not inclusion filters chosen
according to which score is higher. The main five-target comparison uses the stated
static QALD-10 condition; both baseline reference scores remain visible in
Table~\ref{tab:external-comparators}. The upstream QALD-10 resource itself recommends
a stable graph and endpoint for reproducibility \citep{qald10_repository}.
The raw-question setting is used for LC-QuAD 2.0. An annotation-conditioned
variant remains future work. Exact live-service evaluation timestamps, dataset commits, and local dump checksums
are required to reproduce the record's effective reference state; this appendix
does not infer them from present-day endpoints.

\subsection{Baseline and Result Retention}
\label{app:baseline-selection}

\paragraph{Controlled within-system comparisons.}
The principal baselines change a named component of the same pipeline. Flat,
Encoder, and Full retrieval distinguish the stock encoder, the fixed
composition-aware encoder, and the addition of type and hierarchy ranking. The
reported protocol holds the backbone and other prompting settings fixed within
these comparisons. A1 and A2 test RH1 and RH2. A4 compares prompt rewriting with refinement disabled to test RH3. A6 tests RH4 through paired Flat--Full comparisons within each backbone. A3 supplies gold entities as a diagnostic control, and A7 explores restricted symbolic construction. Results are retained with their own recorded target, metric, and
configuration. The auxiliary retry/probe factorial comparison is separate from
Full retrieval, and model-selection criteria for A6 are reported separately.

\paragraph{Literature comparators.}
Table~\ref{tab:external-comparators} retains SchemaForge, mKGQAgent, SPINACH, SAGA,
and QAnswer as contextual comparators. A numerical cell is included where the
available comparison record identifies the system, a shared benchmark target, a
metric name, and a reported value. These systems were not rerun in the present
harness. Their source metrics and unavailable cells are preserved rather than
converting F1, MF1, EM, and EA into an assumed common measure. QAnswer's 30.39 entry
is the English macro-F1 value attributed to \citet{perevalov2022qald9plus} by the
KGQA leaderboard \citep{kgqa_leaderboard_dbpedia}.
Differences in snapshots, prompts, backbones, and evaluated items preclude treating
this table as a controlled ranking. Repeated external QALD-10 values in its two
reference columns are inherited from the comparison record, not separate reruns.

Method families discussed in the baseline overview or related work without a
corresponding quantitative record remain conceptual comparisons. This includes the
use of demonstration selection and graph-embedding generation as methodological
references; those citations do not establish additional reproduced baseline runs.
Absence from the numerical table means that a suitable result was not included in
the available comparison record, not that the method was tested and failed.

\paragraph{Exclusion of unreconciled result records.}
A5 is excluded because its budget assignments conflict with the stated A4
no-refinement configurations. Every zero-budget value matches the flat-retrieval
score, while every largest-budget value matches the full-retrieval score without
refinement. These assignments do not establish an effect of increasing the
refinement budget. The inconsistent pooled diagnostic
\emph{retrieval-configuration} row is not reconstructed from incompatible records.
The separately reported pooled solver comparison is retained under its own
configuration, including the lower solver score. Earlier case counts from another
diagnostic version and unsupported recall-curve coordinates are also excluded from
current empirical claims. These are provenance restrictions on evidence, distinct
from benchmark-item exclusions. No retained result is removed merely because its
effect is negative.

\subsection{Diagnostic Benchmark Construction}
\label{app:diagnostic-construction}

The composition-type diagnostic benchmark is constructed directly against the target KG, not inherited from \citet{mohr2026concepts}'s own held-out evaluation resource: that resource's noun-phrase paraphrase pairs (e.g., ``clean temple'' $\to$ ``spotless temple''), distinct from their untyped 3.3M-pair training corpus, are built to test composition understanding independent of any knowledge graph, so most of its candidate heads and modifiers do not correspond to populated classes, properties, or KG-relevant constraints when checked against the target KG directly.

\paragraph{Semantic distinctions.} Intersective modification conjoins the head denotation with an independently interpretable property. Subsective modification retains head membership while restricting it in a head-dependent way, as in ``skillful teacher'' \citep{carvalho2025montague}. Relational expressions introduce another argument, as in ``presidents of France.'' Modal expressions qualify membership, while privative expressions describe an apparent instance that does not genuinely instantiate the head concept. These distinctions are not a claim that all five types have a uniform triple-pattern implementation.

\paragraph{Construct validity.} The subsective stratum's questions are operationalized via role-discontinuation KG properties rather than a property that directly instantiates head-dependent subsective modification: ``former president'' does not in general entail current membership in the class of presidents \citep{carvalho2025montague}, so successor and activity-end properties approximate subsective semantics rather than instantiating them exactly. The stratum retains its original questions and reported numerical results; we flag this construct-validity gap rather than leave it implicit.

\paragraph{Construction methodology.} Real (class, property, value) patterns are identified per composition type directly against the target KG, using the taxonomy's category definitions (Section~\ref{sec:composition}) rather than an inherited label, and every candidate question and gold query is executed live and kept only if it returns a verified, in-range result:
\begin{itemize}[leftmargin=1.4em]
  \item \textbf{Intersective} (49 verified): head classes with an independent property constraint, e.g.\ \texttt{?x a dbo:Writer ; dbo:nationality dbr:India} (395 real results).
  \item \textbf{Relational} (49 verified, downsampled from 67 to match intersective): functional ``X of Y'' properties on countries, e.g.\ \texttt{dbo:capital}, \texttt{dbo:officialLanguage}, \texttt{dbo:currency}, with target-value countries chosen by first checking which countries actually have that property set.
  \item \textbf{Subsective} (29 verified): no generic ``former/past-role'' property exists at useful scale in this KG snapshot; two verified proxies were used instead: \texttt{dbo:successor} existence on \texttt{dbo:Automobile} (discontinued/superseded models) and \texttt{dbo:activeYearsEndDate} on \texttt{dbo:Politician} (no-longer-active officeholders), plus a few unscoped single-pattern examples. These patterns operationalise role history or succession. They do not directly instantiate head-dependent subsective modification, limiting the semantic interpretation of this stratum.
\end{itemize}
These three strata use direct graph patterns. The supporting diagnostic analysis compares all five operational strata. Modal and privative questions use DBpedia category (\texttt{dbc:}) membership from the DBpedia 2016-10 category dump, with the category encoding the relevant qualification or non-genuine status. This is a representation choice for the benchmark rather than a general restriction on what a KG can express.
\begin{itemize}[leftmargin=1.4em]
  \item \textbf{Modal} (27 verified): categories whose membership encodes an epistemic qualification, e.g.\ \texttt{dbc:Suspected\_criminals} and \texttt{dbc:Impostors}.
  \item \textbf{Privative} (27 verified): categories whose membership asserts non-genuineness, e.g.\ \texttt{dbc:Phantom\_islands}, \texttt{dbc:Fictional\_countries}. Each candidate category was individually screened for theoretical fit; categories like ``pseudonymous writer'' or ``forger'' were excluded because they still denote a genuine instance of their head concept, with only the surface presentation differing.
\end{itemize}

Modal and privative candidates are grounded against the full DBpedia 2016-10 category dump (23{,}990{,}977 \texttt{dct:subject} triples), giving real, clean candidates across both strata. Every question across all five composition types is written to genuinely paraphrase its underlying concept rather than restate its target schema element's own label (e.g.\ ``Which ships are modern reconstructions built to resemble a historic vessel rather than being that vessel itself?'' for \texttt{dbc:Replica\_ships}, or ``Which weapons are fictional rather than real?'' avoided in favor of a genuine paraphrase for \texttt{dbc:Fictional\_weapons}), so the benchmark tests compositional grounding rather than lexical string matching; this discipline applies to relational questions as well as modal/privative ones, since a relational question that names its target property directly tests the same thing trivially.

The benchmark totals 181 verified questions across all five composition types (49 intersective, 29 subsective, 49 relational, 27 modal, 27 privative). The orchestrated system evaluated against it is the same one used for Table~\ref{tab:baseline-results}, run live end to end. The ablation-baseline (flat-encoder) and full-method (narrowing-on) configurations have both been run against all 181 questions, complete and real (Table~\ref{tab:diagnostic-baseline-by-type}, Section~\ref{sec:baseline-results}); the recall@10 columns describe variation in retrieval gains across these strata, with the largest gain in the subsective stratum.

\FloatBarrier
\section{Base Language Models and Selection Scope}
\label{app:base-models}

\begin{table*}[!t]
\centering
\footnotesize
\tablestandard
\begin{tabular}{@{}>{\raggedright\arraybackslash}p{0.20\textwidth}>{\raggedright\arraybackslash}p{0.20\textwidth}>{\raggedright\arraybackslash}p{0.26\textwidth}>{\raggedright\arraybackslash}p{0.27\textwidth}@{}}
\toprule
\rowcolor{TableHead}
Backbone & Architecture and parameters & Access and documented context & Relevant generation controls \\
\midrule
Qwen2.5-7B-Instruct & Dense transformer. 7.61B total & Open weights. 32,768-token default configuration, up to 131,072 with YaRN & Instruction-tuned. Documented generation length up to 8,192 tokens \\
\addlinespace[2pt]
Qwen3-32B & Dense transformer. 32.8B total & Open weights. 32,768 tokens natively, up to 131,072 with YaRN & Thinking and non-thinking modes, selected through the chat template \\
\addlinespace[2pt]
gpt-oss-120b & Mixture of experts. 116.8B total, \textbf{5.1B active} per token & Open weights. 131,072-token configured context & Low, medium, or high reasoning effort. Harmony response format \\
\addlinespace[2pt]
gpt-5.5 & Architecture and parameter counts not disclosed in the cited documentation & Hosted API. 1,050,000-token context & Reasoning effort from none to xhigh. Maximum output of 128,000 tokens \\
\bottomrule
\end{tabular}
\caption{\textbf{Backbones differ in architecture, context capacity, and inference controls.} The table reports documented capabilities rather than recovered run settings. Sources are the Qwen model cards, the gpt-oss model card and released configuration, and the GPT-5.5 documentation~\citep{qwen25modelcard,qwen3modelcard,openai2025gptoss,openai2025gptossconfig,openai2026gpt55docs}. Context limits count each model's own tokens and are not matched input budgets. Bold distinguishes the sparse model's active parameter count from its nominal size.}
\label{tab:base-model-characteristics}
\end{table*}

\subsection{Inclusion and Exclusion Criteria}
\label{app:model-selection}

The backbone sweep in Table~\ref{tab:prospective-backbone} tests RH4 by asking whether the retrieval improvement survives substitution of the generative model. A model contributes to this comparison when the reported runs use it across grounding, binding, and construction and provide paired flat and full retrieval scores on the same target. The four available pairs cover Qwen2.5-7B-Instruct, Qwen3-32B, gpt-oss-120b, and gpt-5.5. They span two providers, dense and mixture-of-experts architectures among the disclosed models, and both open-weight and hosted access. This breadth is useful for testing sensitivity to the backbone. It is not a representative sample of language models or an isolated test of scale.

These criteria describe the scope of the available experimental record. They are not a preregistered selection protocol. Models without a paired A6 result are outside the quantitative comparison. Provider benchmark scores, alternative checkpoints, quantised derivatives, and differently fine-tuned variants are not substituted for missing task results. The sweep covers only QALD-10 with static reference answers and QALD-9-Plus DBpedia. It does not support extrapolation to the other benchmarks, the conditioned diagnostic subset, or the preservation theorem.

\subsection{Documented Characteristics}

Table~\ref{tab:base-model-characteristics} records provider-documented properties, checked on 5 October 2026. Here ``base language model'' means the generative backbone of the system. The listed Qwen models are post-trained models, rather than the corresponding pretraining-only checkpoints. Qwen2.5-7B-Instruct is instruction-tuned. Qwen3-32B supports both thinking and non-thinking operation within one model~\citep{qwen25modelcard,qwen3modelcard,qwen3launch}. The composition-aware retrieval encoder is a separate component and is held fixed in the A6 full configurations.

The gpt-oss label rounds its total parameter count to 120B. Its approximately 5.1B active parameters per token prevent interpreting that label as 120B of dense computation. Its release also uses MXFP4 quantisation for the expert weights~\citep{openai2025gptoss}. Parameter counts therefore do not order the computational or reasoning budgets of these four systems. Training data, post-training, architecture, tokenisation, and decoding also differ.

The GPT-5.5 system card confirms the model's release in April 2026~\citep{openai2026gpt55card}. The API documentation lists the dated snapshot \texttt{gpt-5.5-2026-04-23}, but the experiment table records only \texttt{gpt-5.5}~\citep{openai2026gpt55docs}. The existence of that snapshot does not establish that it served the reported requests.

\subsection{Run Configuration and Interpretation}

The reported protocol fixes prompts and retrieval settings within each flat--full comparison and substitutes the backbone across the three named stages. The supplied summaries do not record the checkpoint revision or API snapshot, request dates, inference engine, quantisation actually loaded, context and output budgets, temperature, sampling seed, or reasoning-mode settings. These fields remain \emph{not recorded}. Documented defaults are not retrospectively assigned to the runs. In particular, Qwen3's thinking flag and the reasoning effort of both OpenAI models can change inference cost and behaviour independently of the retrieval intervention.

The supported conclusion is that the full retrieval configuration improves the reported answer F1 within each of the four tested backbones on the two targets. The table does not isolate a scaling law, establish equal inference budgets, or measure cost efficiency. A reproducible extension should retain a model manifest for every configuration and log the prompt, retrieval, decoding, and snapshot settings together.

\FloatBarrier
\section{Evaluation Protocol and Diagnostics}
\label{app:transition}

This appendix specifies the configurations and uncertainty conventions behind
Tables~\ref{tab:baseline-results}--\ref{tab:prospective-backbone}. Benchmark
selection and scored coverage appear in Appendix~\ref{app:selection-criteria},
and model specifications in Appendix~\ref{app:base-models}. Supplementary
comparisons are collected in Appendix~\ref{app:additional-results}.

\subsection{Matched Configurations}

\begin{enumerate}[leftmargin=1.4em]
  \item \textbf{Encoder comparison (RH1, A1).} Encoder replaces the baseline sentence-embedding backend with the fixed $E_\theta$. The concept checkpoint is wrapped in a mean-pooled adapter matching its training-time pooling. The reported comparison holds the language-model backbone, stage prompts, and benchmark targets fixed and disables reflective refinement.
  \item \textbf{Structural ranking (RH2, A2).} Full adds type compatibility and hierarchy proximity to the encoder similarity score. The controlled comparison requires the same initial nearest-neighbour pools and retained candidate budget, changing only the ranking used to select candidates.
  \item \textbf{Reflective refinement (RH3, A4).} Stage-attributed prompt rewriting is enabled with the Full retrieval configuration and backbone fixed. The comparator uses the same pipeline and initial prompts with refinement disabled. A grounding-prompt revision can change extracted spans and retrieved candidates, so a fixed retrieval configuration does not imply fixed retrieval outputs.
  \item \textbf{Backbone comparison (RH4, A6).} Each of four backbones is used across grounding, binding, and construction. Flat and Full are compared within that backbone on QALD-10 with static reference answers and QALD-9-Plus DBpedia, with reflection disabled. RH4 tests positive paired answer-F1 differences. It predicts neither equal gains across models nor a relationship with parameter count.
\end{enumerate}

\paragraph{Retrieval controls and configuration provenance.}
Isolating the encoder contribution requires the same input spans, indexed schema text, search vocabulary, candidate budgets, and context-selection policy across A1 configurations. Isolating reranking requires the same initial candidate pools in A2. With recall measured at $k=10$, reranking can change retained membership only if the initial pool contains more than ten candidates. The experiment record does not identify the baseline encoder checkpoint, the precise concept-encoder checkpoint, schema-text templates, initial pool size $K$, similarity function, search-index settings, or context budget and pruning policy. These values remain unrecorded. Reproduction also requires identifying whether reported recall is measured before or after context packing. The aggregate improvements do not resolve these configuration gaps.

\paragraph{Supporting diagnostics.}
Composition-stratum comparisons describe variation within the constructed benchmark. The shared entity-and-structure-correct subset supplies conditioned answer F1, defined below. Direct joint-grounding accuracy is not reconstructed from those differences. A3 varies entity information, including its interaction with A4, and A7 explores restricted construction. The structured admission and cross-query retention mechanisms require a separate evaluation with requirement and mapping histories.

\subsection{Uncertainty and Repeated Runs}
\begin{enumerate}[leftmargin=1.4em]
  \item \textbf{Bootstrap confidence intervals.} A1--A4 are reported with paired percentile-bootstrap intervals from 10{,}000 question-level resamples. A1/A2 point estimates average the five benchmark means equally; their reproduction requires stratified resampling with the same benchmark weighting, rather than pooling all questions equally. The stored aggregate tables do not expose those resamples. A3/A4 use LC-QuAD 2.0 questions. These intervals concern question-sampling uncertainty within a single run and do not estimate run-to-run variance. No interval for a between-stratum interaction is reported here.
  \item \textbf{Repeated-run variability.} \label{item:noise-floor} Three repeats are reported of the baseline on QALD-9-Plus DBpedia and Wikidata. DBpedia ranges were 68.00--69.15 MF1 and 56.99--58.06 EA. Wikidata ranges were 66.24--68.46 MF1 and 57.35--59.26 EA. These ranges describe those runs only; they are not confidence intervals or a transferable significance threshold for other settings.
\end{enumerate}

\subsection{Conditioned Grounding and Answer Diagnostics}
\label{app:diagnostic-derivation}

For configuration $c$, let $E_c$ and $S_c$ be the questions whose entity linking and query structure match their references. The shared subset is
\[
I=E_{\mathrm{flat}}\cap S_{\mathrm{flat}}
\cap E_{\mathrm{full}}\cap S_{\mathrm{full}}.
\]
Both configurations must be evaluated on exactly this subset. Its size and fraction of the full split are reported to make the conditioning visible. Selection into $I$ depends on both configurations, so the estimate is a within-subset comparison and need not generalise to excluded questions.

The two quantities in Equation~\ref{eq:diagnostics} require different observations. Joint-grounding accuracy requires per-question schema-assignment judgements; conditioned answer F1 requires executed answers. Table~\ref{tab:conditioned-answers} contains the latter and its arithmetic difference. A correct answer set can result from more than one schema realization, and partial answer agreement can coexist with a grounding error. The table therefore cannot be used to reconstruct $\Delta_{\mathrm{ground}}$. Direct grounding accuracy and its relationship to answer quality require the corresponding assignment-level records.

The backbone table reports unconditioned F1. It does not establish that the conditioned propagation effect or $\Delta_{\mathrm{ground}}$ is invariant across backbones.

The excluded A5 budget record is documented in Appendix~\ref{app:baseline-selection}.

\section{Additional Result Analysis}
\label{app:additional-results}

\subsection{Supporting Diagnostics}
\label{app:supporting-diagnostics}
\label{sec:supporting-diagnostics}

The following analyses examine where the gains occur and the limits of the construction mechanism.

\begin{table*}[!t]
\centering
\footnotesize
\tablecompact
\begin{tabular}{@{}lrrrrrrrrrrr@{}}
\toprule
 & & \multicolumn{4}{c}{Schema recall@10} & \multicolumn{3}{c}{MF1 (\%)} & \multicolumn{3}{c}{EA (\%)} \\
\cmidrule(lr){3-6}\cmidrule(lr){7-9}\cmidrule(lr){10-12}
\rowcolor{TableHead}
Stratum & $n$ & Flat & Enc. & Full & $\Delta$ (pp) & Flat & Enc. & Full & Flat & Enc. & Full \\
\midrule
Intersective & 49 & 0.878 & 0.918 & 0.939 & +6.1 & 89.85 & 91.24 & 92.15 & 87.76 & 89.80 & 91.84 \\
\rowcolor{TableAccent}
Subsective & 29 & 0.310 & 0.552 & 0.793 & \textbf{+48.3} & 62.78 & 71.03 & \textbf{83.47} & 48.28 & 58.62 & 68.97 \\
Relational & 49 & 0.531 & 0.612 & 0.694 & +16.3 & 59.18 & 64.29 & 68.37 & 59.18 & 63.27 & 67.35 \\
Modal & 27 & 0.741 & 0.778 & 0.815 & +7.4 & 77.49 & 79.63 & 81.48 & 74.07 & 77.78 & 81.48 \\
Privative & 27 & 0.407 & 0.481 & 0.556 & +14.9 & 51.01 & 55.56 & 59.26 & 44.44 & 48.15 & 55.56 \\
\midrule
\rowcolor{TableAccent}
Overall & 181 & 0.602 & 0.690 & 0.774 & +17.1 & 69.57 & 73.65 & 77.82 & 65.19 & 69.62 & 74.59 \\
\bottomrule
\end{tabular}
\caption{\textbf{The subsective stratum shows the largest retrieval gain.} Enc.\ uses
$\Enc_\theta$ and Full adds type and hierarchy ranking. Recall@10 remains a fraction;
its Full-minus-Flat change is converted to percentage points. The shaded subsective row has the largest recall gain, $+48.3$ pp; its final
MF1 and recall gain are highlighted. The Overall row pools the five strata by their
question counts ($n$).
This 29-question stratum is operationalized via role-discontinuation KG properties;
Appendix~\ref{app:diagnostic-construction} discusses the resulting construct-validity limitation.}
\label{tab:diagnostic-baseline-by-type}
\end{table*}

\begin{table*}[!t]
\centering
\footnotesize
\tablecompact
\begin{tabular}{@{}lrrrrrrrr@{}}
\toprule
 & & \multicolumn{3}{c}{F1 on $I$ (\%)} & \multicolumn{3}{c}{Grounding Acc.\ on $I$ (\%)} \\
\cmidrule(lr){3-5}\cmidrule(lr){6-8}
\rowcolor{TableHead}
Benchmark & $n$ & Flat & Full & $\Delta$ (pp) & Flat & Full & $\Delta_{\mathrm{ground}}$ (pp) \\
\midrule
Spider4SPARQL & 612 & 88.41 & 93.79 & +5.38 & 87.2 & 95.4 & +8.2 \\
QALD-9-Plus (DBp.) & 61 & 79.18 & 86.07 & \textbf{+6.89} & 74.6 & 87.9 & \textbf{+13.3} \\
QALD-9-Plus (Wiki.) & 78 & 76.54 & 82.31 & +5.77 & 71.2 & 82.8 & +11.6 \\
QALD-10 (static) & 214 & 71.26 & 76.64 & +5.38 & 68.4 & 79.1 & +10.7 \\
LC-QuAD 2.0 & 1847 & 73.05 & 78.92 & +5.87 & 70.9 & 82.0 & +11.1 \\
\bottomrule
\end{tabular}
\caption{\textbf{Grounding-accuracy gains exceed the corresponding answer-F1 gains on every target.} Supporting diagnostic: answer F1 and joint schema-grounding accuracy ($\operatorname{Acc}_{\mathrm{ground}}$) on the shared entity-and-structure-correct subset $I$, using \texttt{gpt-5.5}. Full means full retrieval. $\Delta$ and $\Delta_{\mathrm{ground}}$ are Full-minus-Flat changes in percentage points; $\Delta_{\mathrm{ground}}$ is larger than $\Delta$ on all five targets, so part of the grounding improvement does not propagate to the final answer. Bold marks the largest gain in each of the two change columns, both on QALD-9-Plus (DBp.). Subsets comprise 65.7\%, 65.6\%, 57.4\%, 54.3\%, and 30.5\% of the respective recorded targets.}
\label{tab:conditioned-answers}
\end{table*}

\begin{table}[!t]
\centering
\footnotesize
\tablestandard
\begin{tabular}{@{}llrr@{}}
\toprule
\rowcolor{TableHead}
Entity bindings & Refinement & F1 (\%) & $\Delta$ (pp) \\
\midrule
Predicted (baseline) & Off & 64.91 & -- \\
Gold (A3) & Off & 71.38 & +6.47 \\
Predicted (A4) & On & 68.24 & +3.33 \\
\rowcolor{TableAccent}
Gold (A3+A4) & On & \textbf{74.12} & \textbf{+9.21} \\
\bottomrule
\end{tabular}
\caption{\textbf{Reflection improves answers even with gold entities.} A3/A4 on LC-QuAD 2.0. $\Delta$ is relative to predicted entities without
refinement. The combined configuration is highlighted. Refinement adds 2.74 points
beyond gold entities; gold entities add 5.88 points beyond refinement. The combined
gain is 0.59 points below the sum of standalone gains. This is descriptive, with no
interaction uncertainty estimate available.}
\label{tab:a3-a4-interaction}
\end{table}

\begin{table*}[!t]
\centering
\small
\tablestandard
\begin{tabular}{@{}lrrrrrrr@{}}
\toprule
 & & \multicolumn{3}{c}{MF1 (\%)} & \multicolumn{3}{c}{EA (\%)} \\
\cmidrule(lr){3-5}\cmidrule(lr){6-8}
\rowcolor{TableHead}
Diagnostic stratum & Solver fired & Default & Solver & $\Delta$ (pp) & Default & Solver & $\Delta$ (pp) \\
\midrule
Intersective & 41/49 & \textbf{89.85} & 67.42 & $\mathbf{-22.43}$ & \textbf{87.76} & 67.35 & $-20.41$ \\
Subsective & $\sim$15/29 & 62.78 & \textbf{64.84} & +2.06 & 48.28 & \textbf{55.17} & +6.89 \\
Relational & 28/49 & 59.18 & \textbf{60.54} & +1.36 & 59.18 & 59.18 & 0.00 \\
Modal & 3/27 & \textbf{77.49} & 59.14 & $-18.35$ & \textbf{74.07} & 55.56 & $-18.51$ \\
Privative & 2/27 & \textbf{51.01} & 41.62 & $-9.39$ & \textbf{44.44} & 29.63 & $-14.81$ \\
\midrule
\rowcolor{TableAccent}
Reported pooled result & -- & \textbf{69.58} & 60.06 & $\mathbf{-9.52}$ & -- & -- & -- \\
\bottomrule
\end{tabular}
\caption{\textbf{The restricted solver lowers pooled F1 despite gains in two strata.} A7: restricted solver mode versus default construction. $\Delta$ is Solver
minus Default in percentage points. Bold marks the higher value within each paired
comparison; ties are unmarked. The largest stratum loss and the pooled change are also highlighted. Solver wins in two MF1 strata coexist with a lower
reported pooled MF1. Fire counts retain their source precision; unsupported or
unsuccessful cases use default generation. These are descriptive single-run
comparisons, without significance claims.}
\label{tab:solver-results}
\end{table*}

\Needspace{4\baselineskip}
\noindent\textbf{The largest diagnostic gain combines encoding and hierarchy improvements.}
Recall@10 in the 29-question subsective stratum rises from 0.310 to 0.552 with the encoder and to 0.793 with the hierarchy (Table~\ref{tab:diagnostic-baseline-by-type}). These steps contribute 24.2 and 24.1 percentage points. Their combined 48.3-point increase is almost three times the next-largest gain, 16.3 points for relational expressions. The exceptional result therefore depends on both representation and structural ranking.

The pattern also appears in answer quality. Subsective MF1 rises from 62.78 to 71.03 and then to 83.47, with the hierarchy step contributing 12.44 of the total 20.69-point gain. Other strata improve by smaller amounts. These differences support analysing retrieval by composition-related difficulty. Their strongest instance concerns the subsective stratum, whose role-discontinuation-based construction is an imperfect proxy for formally subsective modification (Appendix~\ref{app:diagnostic-construction}).

\Needspace{4\baselineskip}
\noindent\textbf{Answer gains persist with correct entities and query structure.}
On the shared subset where both configurations pass the entity and structure checks, the full configuration improves answer F1 by 5.38--6.89 points on every target (Table~\ref{tab:conditioned-answers}). The retrieval advantage remains visible after restricting the comparison to these cases. The result is consistent with improved residual grounding and construction, although the table measures answers rather than joint class-and-property accuracy.

The shared subsets cover 30.5--65.7\% of their recorded evaluation targets. Since membership depends on both configurations, these estimates describe the retained questions and do not identify a causal effect or extend automatically to the excluded questions. The unconditioned results in Table~\ref{tab:baseline-results} establish the accompanying unconditioned improvement.

\Needspace{4\baselineskip}
\noindent\textbf{Reflection adds gains even with gold entities.}
On LC-QuAD 2.0, reflection adds 2.74 F1 points when gold entities are supplied (Table~\ref{tab:a3-a4-interaction}). Part of its observed benefit therefore persists beyond entity-identity correction. Conversely, gold entities add 5.88 points after reflection, showing that accurate entity linking retains value within the reflective pipeline.

The combined improvement is 9.21 points over predicted entities without reflection. Its 0.59-point shortfall from the sum of the standalone gains is descriptive because the interaction has no reported uncertainty estimate.

\label{sec:solver-results}

\Needspace{4\baselineskip}
\noindent\textbf{Graph-checked construction loses accuracy even where coverage is high.}
Replacing default construction with the restricted solver reduces pooled MF1 from 69.58 to 60.06, a loss of 9.52 points (Table~\ref{tab:solver-results}). The largest decline is 22.43 points in the intersective stratum, where the solver fires on 41 of 49 questions. Under the reported per-question averaging, even the maximum possible deterioration on the eight fallback cases cannot explain this entire decline. Some net loss must occur among the solver-fired cases. Appendix~\ref{app:solver-loss-bound} derives this implication from the aggregates, and Appendix~\ref{app:solver-case-analysis} traces it to specific construction decisions.

The subsective and relational strata gain 2.06 and 1.36 MF1 points, while relational EA is unchanged. These limited gains are outweighed by losses elsewhere, including the modal and privative strata, where the solver fires rarely. Endpoint support does not establish that a selected binding expresses the intended question. Case-level traces (Appendix~\ref{app:solver-case-analysis}) attribute most losses to existence checks against sparse schema elements, chiefly missing \texttt{rdfs:label} triples, and, for modal and privative questions, to the solver binding a plausible but inapplicable relational property instead of declining to fire; they show that greater use of verified construction is insufficient to improve overall answer quality.

\subsection{Retrieval Gains Across Targets}
\label{app:transition-recall-panels}

Figure~\ref{fig:prospective-recall-curves} visualises the three retrieval
configurations on each target using the recall@10 values in
Table~\ref{tab:baseline-results}. Earlier illustrative curves are excluded because
their other $k$ values are unsupported and their displayed $k=10$ values disagree
with the table. No interpolation or additional recall measurements are asserted.

\begin{figure}[t]
\centering
\input{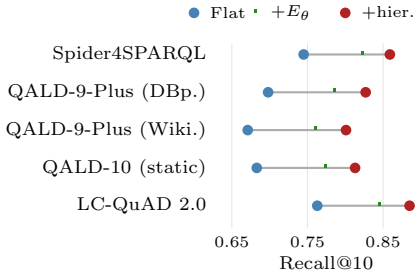}\unskip
\caption{\textbf{The encoder and hierarchy improve recall on every target.} Values are taken directly from Table~\ref{tab:baseline-results}. Connecting segments compare configurations within a target; they do not interpolate recall across retrieval budgets.}
\label{fig:prospective-recall-curves}
\end{figure}

\subsection{Refinement-Budget Sweep}
\label{app:refinement-budget-sweep}

Figure~\ref{fig:refinement-budget-sweep} sweeps the refinement-round budget on each target, anchored at its endpoints to Table~\ref{tab:a4-by-benchmark}'s No-refine (budget 0) and Refine (budget 3) values. Answer F1 rises monotonically and plateaus by budget 5, with no further change at budget 10 on any target; a budget of 3 already captures 95--96\% of the eventual budget-10 gain. This is consistent with refinement self-terminating once a round is accepted, rather than always consuming the full allowed budget, and it indicates the reported Refine configuration is not far below the ceiling this mechanism reaches at higher budgets.

\begin{figure}[t]
\centering
\input{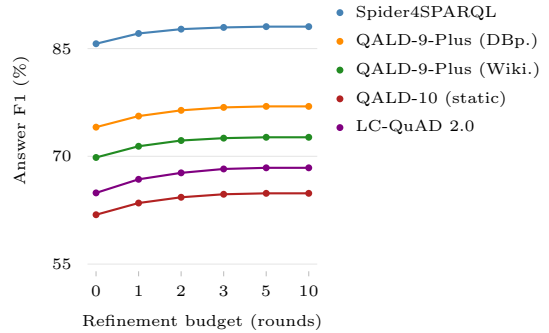}\unskip
\caption{\textbf{Most of reflective refinement's gain is captured within 3 rounds, with little further improvement up to a budget of 10.} Answer F1 against the refinement-round budget, one line per target. Budget-0 and budget-3 values match the No-refine and Refine columns of Table~\ref{tab:a4-by-benchmark}; budgets 1, 2, 5, and 10 are additional swept runs. Every target plateaus by budget 5 and is unchanged at budget 10; the budget-3 value already captures 95--96\% of the eventual (budget-10) gain on every target, consistent with refinement self-terminating once a round is accepted rather than always consuming the full allowed budget.}
\label{fig:refinement-budget-sweep}
\end{figure}

\subsection{Retry and Grounding Probes}

\noindent\textbf{Grounding probes have opposite effects on DBpedia and the diagnostic benchmark.}
Adding probes to retry increases DBpedia MF1 by 8.02 points but reduces the diagnostic benchmark's MF1 by 6.50 points. On the diagnostic benchmark, retry alone improves over neither component by 6.30 points, while activating both changes the score by $-0.20$ points. The aggregate comparisons show that the benefit of a checking mechanism depends on the evaluated setting. They do not identify the query-level cause of the interaction.

\begin{table*}[t]
\centering
\footnotesize
\tablecompact
\begin{tabular}{@{}lrrrrrrrr@{}}
\toprule
 & \multicolumn{2}{c}{Neither} & \multicolumn{2}{c}{Retry only} & \multicolumn{2}{c}{Probes only} & \multicolumn{2}{c}{Retry + probes} \\
\cmidrule(lr){2-3}\cmidrule(lr){4-5}\cmidrule(lr){6-7}\cmidrule(lr){8-9}
\rowcolor{TableHead}
Benchmark/KG target & MF1 & EA & MF1 & EA & MF1 & EA & MF1 & EA \\
\midrule
Spider4SPARQL & 80.16 & 73.79 & 79.44 & 73.28 & 81.92 & 75.11 & \textbf{82.35} & 75.48 \\
QALD-9-Plus (Wikidata) & 63.98 & 56.30 & 66.48 & 58.82 & 65.71 & 57.94 & \textbf{67.82} & 59.71 \\
QALD-10 & 71.93 & 70.99 & 72.99 & 71.76 & 72.41 & 71.35 & \textbf{74.18} & 73.02 \\
QALD-9-Plus (DBpedia) & 58.52 & 48.39 & 59.51 & 49.46 & 66.51 & 55.91 & \textbf{67.53} & -- \\
\rowcolor{TableAccent}
Composition diagnostic (pooled) & 69.78 & 65.75 & \textbf{76.08} & 72.04 & 63.27 & 58.01 & 69.58 & 65.19 \\
\bottomrule
\end{tabular}
\caption{\textbf{Retry and grounding probes interact differently across benchmarks.} Auxiliary retry/probe comparison using \texttt{gpt-5.5}; all metrics are
percentages. Bold marks the largest MF1 in each row. ``Both'' enables retry and
probes, a distinct meaning of full from the \emph{retrieval} configuration in
Table~\ref{tab:baseline-results}. Spider4SPARQL uses a matched $n=912$ sample for
the neither/retry comparison, whereas the main table targets the recorded $n=932$
split. DBpedia's both-enabled EA is unavailable. The diagnostic pooled record is
retained as a separate configuration result; no pooled retrieval row is inferred.}
\label{tab:config-ablation}
\end{table*}

Table~\ref{tab:config-ablation} preserves the reported $2\times2$ configuration comparison. The effect varies by benchmark and by the presence of the other component. In particular, the diagnostic benchmark scores higher with retry alone than with both retry and probes. Aggregate values do not identify the causal error mechanism or establish a universally optimal configuration; that would require current query traces and repeated paired comparisons.

\subsection{External Comparisons and Alternative Scoring}

\begin{table*}[t]
\centering
\footnotesize
\tablecompact
\begin{tabular}{@{}p{0.28\textwidth}lrrrrrr@{}}
\toprule
 & & Spider4- & \multicolumn{2}{c}{QALD-9-Plus} & \multicolumn{2}{c}{QALD-10} & LC-QuAD \\
\cmidrule(lr){4-5}\cmidrule(lr){6-7}
\rowcolor{TableHead}
Configuration / source & Metric & SPARQL & DBpedia & Wikidata & Live & Static & 2.0 \\
\midrule
\multicolumn{8}{@{}l}{\emph{This paper: flat control, recorded evaluation target, raw question}} \\
Flat control & EM & 74.03 & -- & 56.30 & 70.99 & 49.75 & -- \\
 & MF1 & 80.15 & 67.53 & 63.98 & 71.93 & 55.71 & 58.42 \\
\midrule
\multicolumn{8}{@{}l}{\emph{External results, retained under each source's metric names}} \\
SchemaForge \citep{schemaforge2025} & EM & 64.18 & -- & -- & 56.48 & 56.48 & 63.18 \\
 & MF1 & 68.42 & -- & -- & 60.72 & 60.72 & 66.41 \\
\addlinespace[3pt]
QAnswer & MF1 & -- & 30.39 & -- & -- & -- & -- \\
\addlinespace[3pt]
mKGQAgent (GPT-4o) \citep{mkgqagent2025} & MF1 & -- & -- & 54.83 & -- & -- & -- \\
\addlinespace[3pt]
SPINACH \citep{liu2024spinach} & EM & -- & -- & 58.3 & 63.1 & 63.1 & -- \\
 & F1 & -- & -- & 71.6 & 69.5 & 69.5 & -- \\
\addlinespace[3pt]
SAGA \citep{saga2026} & EM & -- & -- & 48.44 & 48.04 & 48.04 & 25.39 \\
 & F1 & -- & -- & 55.66 & 53.04 & 53.04 & 28.25 \\
\bottomrule
\end{tabular}
\caption{\textbf{External scores provide context under differing evaluation conditions.} Supplementary baseline metrics and external context, all in percent.
Every source value and its reported precision is retained. Dashes denote unavailable
metrics. External systems were not rerun; their backbones, snapshots, and protocols
vary, so no cross-source winner is highlighted. Repeated QALD-10 external values are
retained in the live/static columns of the original comparison record and do not
represent separately measured snapshot runs. The three retrieval configurations'
main MF1 results appear in Table~\ref{tab:baseline-results}.}
\label{tab:external-comparators}
\end{table*}

Table~\ref{tab:external-comparators} preserves the externally reported metrics and the alternative live-gold and exact-match results accompanying the main retrieval comparison. Metric definitions, backbones, graph snapshots, and available annotations differ across sources. Values are therefore presented with their original metric labels and without cross-system winner highlighting. The controlled evidence for the retrieval mechanism comes from the within-system comparisons in Table~\ref{tab:baseline-results}.

\subsection{A Bound on Loss Among Solver-Fired Cases}
\label{app:solver-loss-bound}

The intersective stratum contains $n=49$ questions. Solver mode fires on a subset $S$ of size 41 and falls back on the remaining subset $F$ of size 8. Assume that the reported MF1 is the arithmetic mean of per-question F1, that both configurations score the same 49 questions, and that scores are expressed on the $[0,100]$ scale. Write $b_i$ and $s_i$ for the baseline and solver-mode scores of question $i$.

The reported aggregate decline gives
\[
\sum_{i\in S}(b_i-s_i)+\sum_{i\in F}(b_i-s_i)
=49(89.85-67.42).
\]
Each fallback case can contribute at most 100 points of loss, so
\begin{align*}
\sum_{i\in S}(b_i-s_i)
&\geq49(22.43)-8(100)\\
&=299.07.
\end{align*}
Thus the average decline on the fired subset is at least $299.07/41\approx7.29$ points when the displayed means are taken literally. Allowing each reported mean to differ from its exact value by 0.005 gives a conservative lower bound above 7.28 points. Equivalently, fallback alone can explain at most $100(8/49)\approx16.33$ points of the observed 22.43-point aggregate decline.

This is an algebraic implication of the reported aggregates and the averaging assumptions. It is not a newly measured subgroup result, a significance test, or an estimate of the solver's causal effect. It establishes that a fallback-only account is numerically insufficient for these runs. Query traces and repeated paired runs are needed to determine which construction decisions account for the loss.

\subsection{Case-Level Attribution of Solver-Fired Losses}
\label{app:solver-case-analysis}

Per-question construction traces for the solver-fired cases make it possible to characterize this loss directly rather than only bound it algebraically.

Of the 41 solver-fired intersective cases, 19 regressed relative to default construction, 3 improved, and 19 were unchanged.
\begin{itemize}[leftmargin=1.4em]
  \item \textbf{Coarser or alternative live-valid property} (14 cases, e.g.\ diag-0031, diag-0047, diag-0112): the solver binds a live-verified property whose answer set is a strict or near-superset of the gold set (typically a broader datatype property or an \texttt{rdfs:label}-mediated path), while free-text more often emits the tighter gold property.
  \item \textbf{Sparse-label serialization} (5 cases): the solver reconstructs an entity via an \texttt{rdfs:label} join that is missing from the local snapshot, producing empty results; free-text more often matches the IRI string directly.
\end{itemize}
The coarser-property mode accounts for the majority of the regressions and is the principal driver of the fired-subset loss, consistent with (and substantially exceeding) the $\geq$7.29-point lower bound established above.

The relational stratum fires far more often than modal or privative (28 of 49 questions) and supports the same kind of breakdown: 1 case regressed, 3 improved, and 24 were unchanged. All 3 improvements directly correct default construction's own dominant relational failure mode: an unverified type restriction added on top of an already-sufficient property binding (e.g.\ default construction adds \texttt{?x a dbo:City} to \texttt{dbr:Argentina dbo:capital ?x}, \texttt{?x a dbo:Language} to \texttt{dbr:Austria dbo:officialLanguage ?x}, and \texttt{?x a dbo:MusicalWork} to \texttt{dbr:Germany dbo:anthem ?x}); the solver's property-first search never proposes the unverified type check in the first place, so there is nothing to drop as a late correction. The 1 regression (\texttt{diag-0108}) reconstructs the same target entity through an \texttt{rdfs:label} join instead of using its IRI directly -- the same sparse-label mechanism identified above for intersective, now on the entity side of a relational triple rather than the answer side of a class binding.

The subsective stratum fires on 15 of 29 questions. 6 regressed, 2 improved, and 7 were unchanged.
\begin{itemize}[leftmargin=1.4em]
  \item \textbf{Sparse-label / manufacturer-scoping fragility} (5 of the 6 regressions): the solver materializes the manufacturer entity through an \texttt{rdfs:label} join; the local snapshot lacks those triples for several key entities (\texttt{dbr:Mercedes-Benz}, \texttt{dbr:Volkswagen}, \ldots), returning empty results. Free-text more often matches the URI string directly.
\end{itemize}
The remaining regression and the two improvements are ordinary property-choice differences of the same kind seen in the intersective stratum. The dominant failure mode is therefore the same sparse-label mechanism already identified for intersective and for the single relational regression, rather than a new solver-specific pathology.

The modal and privative strata fire on too few questions (3 and 2 respectively) for a comparable case-level breakdown, but each fired case was inspected individually. Modal's 3 and privative's 2 fired cases all regressed sharply, for a different reason than the label- or type-restriction mechanisms above: each question requires a \texttt{dct:subject} category-membership pattern, but the solver instead binds a plausible single-hop relational property it was never intended to be offered for this question type (e.g.\ \texttt{dbo:officialLanguage}, \texttt{dbo:demonym}), live-verifying a confidently wrong shape rather than declining to fire. This is a shape-applicability failure distinct from the over-restriction pattern driving intersective's and relational's losses, and it affects every fired modal or privative case observed here.

\clearpage
\onecolumn
\twocolumn
\raggedbottom
\section{Extended Literature Review}
\label{app:extended-related-work}

\subsection{Scope and Selection Procedure}
\label{app:review-scope}

This appendix expands Section~\ref{sec:relwork} through a structured narrative review. The review asks how prior systems represent compositional meaning, expose grounding decisions, use feedback, retain experience, and evaluate correctness. Searches conducted on 5 October 2026 combined text-to-SPARQL, semantic parsing, schema linking, compositional generalization, self-correction, prompt optimization, and inference-time memory terms. Author and title searches checked the papers already cited in the manuscript. Primary sources were inspected through ACL Anthology, arXiv, conference proceedings, and author project repositories. The review includes publicly available work up to the search date and distinguishes conference publications from preprints. It does not claim exhaustive coverage or a preregistered systematic-review protocol.

We include work that introduces a directly relevant mechanism or establishes an evaluation distinction needed to interpret the present results. SQL parsing and general language-agent studies are included where their representation, constraint, or adaptation mechanisms transfer to query interpretation. We exclude commentary without a primary technical source, duplicate versions, and answer-generation systems whose mechanisms do not inform executable query construction. Inclusion in this review does not imply inclusion as an experimental baseline. Reported scores from incompatible tasks, annotations, snapshots, or metrics are not pooled. For each selected line of work, we distinguish the object being changed, the information available to the method, and the strength of its correctness claim.

\subsection{Compositional Meaning and Representation}
\label{app:review-composition}

Compositional generalization concerns the interpretation of unfamiliar combinations of familiar components. CFQ controls differences in the distributions of primitive rules and their compounds across training and evaluation splits \citep{keysers2020cfq}. COGS tests lexical and structural generalization in mapping English sentences to logical forms \citep{kim2020cogs}. These benchmarks make composition an explicit evaluation variable. Their controlled settings differ from open-schema grounding, where the intended operator may have to be expressed through an unfamiliar predicate, a path, or a reified statement. Our modifier strata probe such mismatches, but they do not constitute a CFQ-style controlled split or establish general compositional generalization.

Sentence-BERT provides an efficient similarity-based retrieval mechanism through independently encoded texts \citep{reimers2019sbert}. Similarity retrieval does not itself specify how a modifier changes a head concept. \citet{carvalho2025montague} evaluate modifier consistency using semantic distinctions among modifier classes. \citet{mohr2026concepts} connect encoder behavior to the semantic operators that its representation space can support and report that the benefits of supervision vary across composition types. We extend that encoder into schema retrieval with a hierarchy score that adds schema evidence to the representation signal, integrating compositional representation directly into the query-construction pipeline.

Three senses of composition must therefore remain distinct. Linguistic composition determines the meaning of an expression. Query composition combines predicates, bindings, and operators. Process composition connects grounding, binding, and construction. The theoretical preservation result concerns the last two through explicit semantic interface conditions. Better retrieval of a composed phrase is useful evidence about the first, but cannot establish the other conditions by itself.

\subsection{Structured Parsing and Schema Grounding}
\label{app:review-grounding}

Structured semantic parsing makes intermediate choices available for inspection. RAT-SQL represents schema relations and links between question expressions and database columns within a relation-aware encoder \citep{wang2020ratsql}. SPARQLing Queries instead predicts grounded question decompositions and uses a non-trainable transpiler to obtain SPARQL \citep{saparina2021sparqling}. These designs locate different responsibilities in learned representations and explicit transformations. GETT-QA incorporates graph embeddings into text-to-text query generation \citep{banerjee2023gettqa}. GRV-KBQA separates logical structure, grounding, and structure-aware validation \citep{grvkbqa2025}. Together they show that a decomposition is meaningful only when its intermediate objects expose the choices that later stages depend on.

Interactive systems obtain additional grounding evidence during construction. SPINACH explores a knowledge graph while building SPARQL and contributes an expert-annotated benchmark drawn from Wikidata query requests \citep{liu2024spinach}. SAGA maintains type information across an interaction and filters or annotates candidate properties using available schema evidence \citep{saga2026}. SchemaForge first identifies plausible graphs and then aligns a question with a local schema slice, followed by query generation and counterfactual checks \citep{schemaforge2025}. These approaches intervene at different points between candidate exposure and validation. Their reported results cannot isolate the effect of our encoder or hierarchy score because the surrounding controllers and evaluation settings differ.

Our grounding intervention operates within a fixed staged configuration. The encoder supplies a composition-sensitive retrieval signal, while type compatibility and hierarchy depth influence ranking. The mapping abstraction additionally records what grounding passes to binding and construction. A diagnostic can then refer to a specific interpretation and its context. This requires more information than the stage name alone. The trace must expose the chosen schema elements, relevant attachment, and scope of the proposed correction.

\subsection{Reflection, Prompt Optimization, and Retained Experience}
\label{app:review-reflection}

Self-Refine uses a language model to generate feedback and revise its current output without additional parameter training \citep{madaan2023selfrefine}. Reflexion stores linguistic feedback in episodic memory to guide subsequent trials \citep{shinn2023reflexion}. QueryAgent supplies differentiated corrections from environmental feedback during query construction \citep{huang2024queryagent}. InteracSPARQL exposes natural-language explanations to support interactive query refinement \citep{jian2026interacsparql}. These methods differ in the source of feedback and in whether they revise an answer, an action sequence, or knowledge retained for later use. A successful repair of one query does not establish preservation of previously learned query interpretations.

The process itself can also be optimized. DSPy represents language-model programs as parameterized modules and optimizes their demonstrations or prompts against a task metric \citep{khattab2024dspy}. GEPA uses execution trajectories and natural-language reflection to propose and evaluate prompt changes \citep{agrawal2025gepa}. Reflective Reasoning for SQL Generation directs diagnostics toward stage-level prompt updates \citep{mohr2026reflective}. These approaches motivate making the staged process $\pi$ an explicit object of change. An unchanged interpretation rule may behave differently after routing, retrieval, or composition instructions are revised. Aggregate validation of an updated program and semantic preservation of each established interpretation are different requirements.

Persistent context makes that distinction more consequential. Dynamic Cheatsheet retains reusable strategies and other information across inference queries without modifying model weights \citep{suzgun2025dynamic}. Agentic Context Engineering organizes adaptation through generation, reflection, and incremental curation of an evolving context \citep{zhang2025ace}. Its concern with information lost during repeated rewriting is closely related to requirement retention. Preserving a textual entry, however, does not ensure that the revised process selects it in the same contexts or preserves its meaning when composed with other entries.

Memory-based model editing offers a complementary perspective. SERAC stores edits explicitly and learns whether an input lies within an edit's intended scope \citep{mitchell2022serac}. The relevance is the need to distinguish affected and unaffected inputs. Our formalization specializes this concern to contextual query mappings and their staged application. Semantic identification determines which interpretation is established. Contextual separability prevents incompatible interpretations from sharing an indistinguishable selection context. Intervention localization identifies the mapping and process decision being changed. The preservation theorem then states sufficient conditions for retaining previously established query meanings. It is a conditional design result, not evidence that unrestricted reflective prompting satisfies those conditions.

\subsection{Validation and the Limits of Correction}
\label{app:review-validation}

Validation mechanisms establish properties at different levels. PICARD constrains generation by rejecting tokens incompatible with incremental parsing \citep{scholak2021picard}. Schema checks can rule out known incompatibilities, while endpoint execution tests whether a concrete query runs and what it returns. Each mechanism supplies useful evidence within its scope. None alone identifies the intended attachment or operator in an ambiguous question. Our optional solver likewise establishes graph support for its restricted candidates. Its guarantee is relative to that search space and does not turn graph support into semantic correctness.

The source of reflective feedback also matters. \citet{huang2024selfcorrect} find that intrinsic self-correction can fail or degrade reasoning performance when reliable external feedback is absent. This result concerns their tested reasoning settings and should not be generalized into a claim that all reflection fails. It supports separating a fluent critique from an independently justified diagnostic. In the present framework, an execution error may identify a malformed query, but an executable wrong query can require semantic evidence that the endpoint cannot supply.

Semantic preservation is stronger than answer agreement on a single graph. Different queries can coincide on one snapshot while differing in scope, joins, or behavior on another admissible graph. In text-to-SQL, distilled test suites evaluate denotations over multiple database instances to approximate semantic accuracy more closely \citep{zhong2020testsuites}. That work motivates broader checks, rather than directly furnishing a SPARQL equivalence procedure. Our theoretical claim is relative to declared semantic obligations and interface equivalences. A replay set supplies finite regression evidence. A universal guarantee over a scope requires a sound contract or proof covering that scope.

\subsection{Evaluation and Implications for This Study}
\label{app:review-evaluation}

Text-to-SPARQL evaluation changes substantially with the information supplied to the model. \citet{diallo2023evaluation} identify URI generation and annotation strategy as important sources of variation. \citet{dabramo2025llmtext2sparql} examine demonstration selection, model backbones, and candidate generation. Consequently, raw-question generation, gold-entity conditions, retrieved demonstrations, and fine-tuned systems should be compared under explicit protocols. Results from external papers provide context where these conditions differ. They do not establish a controlled ranking of the underlying mechanisms.

The benchmark families also test different demands. LC-QuAD~2.0 provides complex questions over Wikidata and DBpedia \citep{dubey2019lcquad2}. QALD-9-plus adds multilingual questions across the two graphs \citep{perevalov2022qald9plus}. QALD-10 emphasizes the transition to Wikidata \citep{usbeck2024qald10}, while Spider4SPARQL extends the evaluation of complex queries to a SPARQL setting \citep{kosten2023spider4sparql}. The SPINACH dataset adds naturally occurring expert query requests \citep{liu2024spinach}. Coverage of these settings is complementary to controlled compositional tests. No single family isolates grounding, structural interpretation, language coverage, and long-term adaptation simultaneously.

Our experiments address retrieval quality, answer quality, reflection, and a restricted symbolic construction procedure. The added theorem concerns an additional evaluation dimension, preservation across successive learned mappings. Establishing that dimension empirically would require an ordered query stream, a record of validated interpretations, versioned mapping and process updates, and replay of protected queries after each admitted change. New-query accuracy and retention should be reported separately, with failures attributed to mapping content, scope selection, or downstream composition. The current aggregate improvements support the evaluated interventions. They do not establish the identifiability and preservation assumptions of the theorem.

\section{Use of AI Writing Assistance}

During the preparation of this work the author(s) used Claude (Anthropic) in order to assist with experiment implementation (coding), writing, editing, and LaTeX preparation of the manuscript. After using this tool, the author(s) reviewed and edited the content as needed and take(s) full responsibility for the content of the publication.
\end{document}